\newif\ifsubmission
\newif\iftqc
\newif\ifextendedabstract

\ifsubmission
    \documentclass{llncs}
\else
    \documentclass[letterpaper,11pt]{article}
    \usepackage{times}
    \usepackage{fullpage}
    \usepackage{amsthm}
\fi
\ifextendedabstract
\fi

\usepackage{amsfonts,amssymb,amsmath}
\usepackage{mathrsfs}

\ifextendedabstract
\else
\usepackage{authblk}
\fi

\usepackage[colorlinks=true,linkcolor=blue,citecolor=blue,breaklinks=true]{hyperref}
\usepackage{breakcites}
\usepackage{algorithm}
\usepackage{algorithmicx}
\usepackage[noend]{algpseudocode}
\usepackage{bm} 
\usepackage{graphicx, float, tikz}
\usepackage{colonequals}  
\usepackage{color}
\usepackage{multicol}
\usepackage{tikz}
\usepackage{fancybox}
\usepackage[breakable, theorems, skins]{tcolorbox}
\usetikzlibrary{matrix,arrows}
\usetikzlibrary{positioning}
\usetikzlibrary{decorations}
\usetikzlibrary{decorations.pathreplacing}
\usetikzlibrary{shapes.geometric, chains, calc, fit, matrix}
\tikzstyle{system}=[rectangle,draw,fill=lightgray,minimum height=0.8cm,minimum
            width=0.8cm,thick]
\tikzstyle{BC}=[system]
\tikzstyle{resource}=[system]
\tikzstyle{RO}=[resource, minimum width=1cm]
\tikzstyle{protocol}=[circle, inner sep=0.7mm, draw]
\tikzstyle{simulator}=[circle, inner sep=0.7mm, draw]
\tikzstyle{memory}=[resource]
\tikzstyle{distinguisher}=[resource,fill=white,minimum width=3.5cm,
minimum height=1.2cm]
\tikzstyle{link}=[]

\usepackage[capitalize]{cleveref}

\usepackage{paralist}
\usepackage{enumerate}
\usepackage{enumitem}
\usepackage[n,
advantage,
operators,
sets,
adversary,
landau ,
probability,
notions,
logic,
ff,
mm,
primitives,
events,
complexity,
asymptotics,
keys]{cryptocode}

\renewcommand{\secpar}{\lambda}
\usepackage{tikz}
\usetikzlibrary{arrows,automata}

\usepackage{braket}

\usepackage[framemethod=tikz]{mdframed}
\usepackage{subcaption}

\usepackage{algpseudocode}

\usepackage{cleveref}
\usepackage{framed}

\newcommand{\personalnote}[3]{
    \expandafter\newcommand\csname#1\endcsname[1]{
        {}  
        }
    }

\personalnote{Justin}{Justin}{blue}
\personalnote{James}{James}{orange}

\newcommand{\bbC}{\mathbb{C}}

\newcommand{\bbE}{\mathbb{E}}
\newcommand{\bbF}{\mathbb{F}}

\newcommand{\bbI}{\mathbb{I}}

\newcommand{\bbK}{\mathbb{K}}

\newcommand{\bbN}{\mathbb{N}}

\newcommand{\bbS}{\mathbb{S}}

\newcommand{\bbZ}{\mathbb{Z}}

\newcommand{\calA}{\mathcal{A}}

\newcommand{\calC}{\mathcal{C}}

\newcommand{\calH}{\mathcal{H}}
\newcommand{\calI}{\mathcal{I}}

\newcommand{\calP}{\mathcal{P}}

\newcommand{\calR}{\mathcal{R}}
\newcommand{\calS}{\mathcal{S}}

\newcommand{\calX}{\mathcal{X}}
\newcommand{\calY}{\mathcal{Y}}

\newcommand{\vect}[1]{\mathbf{#1}}
\newcommand{\Delete}{\mathsf{Delete}}
\newcommand{\Verify}{\mathsf{Verify}}
\newcommand{\Reconstruct}{\mathsf{Reconstruct}}
\newcommand{\Split}{\mathsf{Split}}
\newcommand{\Correct}{\mathsf{Correct}}
\newcommand{\Sim}{\mathsf{Sim}}
\newcommand{\Adv}{\mathsf{Adv}}
\newcommand{\secp}{\lambda}
\newcommand{\CSplit}{\mathsf{CSplit}}
\newcommand{\CReconstruct}{\mathsf{CReconstruct}}
\newcommand{\Twotwo}{2\text{-}2}
\newcommand{\csh}{\mathsf{csh}}
\newcommand{\qsh}{\mathsf{qsh}}
\newcommand{\Interpolate}{\mathsf{Interpolate}}

\newcommand{\verkey}{\mathsf{vk}}
\newcommand{\cert}{\mathsf{cert}}

\newcommand{\Tr}{\mathsf{Tr}}
\newcommand{\TD}{\mathsf{TD}}
\newcommand{\tracedist}{\TD}
\newcommand{\QFT}{\mathsf{QFT}}

\newcommand{\NSCD}{\mathsf{SS\text{-}NSCD}}
\newcommand{\ACD}{\mathsf{SS\text{-}ACD}}
\newcommand{\TwoSSCD}{2\NSCD}
\newcommand{\sh}{\mathsf{sh}}

\newcommand{\chalReg}{\mathcal{C}}

\newcommand{\delLoss}{\ell}
\newcommand{\numChecks}{r}
\newcommand{\delLossValueShort}{t\frac{\log(\secpar)}{\sqrt{\numChecks}}}

\newcommand{\corctr}{a}
\newcommand{\delctr}{b}

\ifsubmission \else
\newtheorem{definition}{Definition}
\newtheorem{lemma}{Lemma}
\newtheorem{theorem}{Theorem}
\newtheorem{claim}{Claim}

\newtheorem{remark}{Remark}
\fi
\newtheorem{fact}{Fact}

\iftqc
\title{Secret Sharing with Certified Deletion\footnote{In parallel submission to CRYPTO 2024, which emphasizes the cryptographic aspect. This work may also be interesting to the quantum computing community.}}
\else 
\title{Secret Sharing with Certified Deletion}
\fi

\ifsubmission
    \author{}
    \institute{}
\else
    \ifextendedabstract
        \author{James Bartusek\thanks{UC Berkeley} \and Justin Raizes\thanks{Carnegie Mellon University}}
    \else
        \author[1]{James Bartusek}
        \author[2]{Justin Raizes}
        \affil[1]{UC Berkeley}
        \affil[2]{Carnegie Mellon University}
    \fi 
\fi
\date{}

\allowdisplaybreaks

\begin{document}

\maketitle

\begin{abstract}

    Secret sharing allows a user to split a secret into many shares so that the secret can be recovered if, and only if, an authorized set of shares is collected. Although secret sharing typically does not require any computational hardness assumptions, its security \emph{does} require that an adversary cannot collect an authorized set of shares. Over long periods of time where an adversary can benefit from multiple data breaches, this may become an unrealistic assumption.

    \iftqc
\else
         We initiate the systematic study of secret sharing \emph{with certified deletion} in order to achieve security \emph{even against an adversary that eventually collects an authorized set of shares}. In secret sharing with certified deletion, a (classical) secret $s$ is split into quantum shares that can be destroyed in a manner verifiable by the dealer.
        
        We put forth two natural definitions of security. \textbf{No-signaling security} roughly requires that if multiple non-communicating adversaries delete sufficiently many shares, then their combined view contains negligible information about $s$, even if the total set of corrupted parties forms an authorized set. \textbf{Adaptive security} requires privacy of $s$ against an adversary that can continuously and adaptively corrupt new shares and delete previously-corrupted shares, as long as the total set of corrupted shares minus deleted shares remains unauthorized.
    
        Next, we show that these security definitions are achievable: we show how to construct (i) a secret sharing scheme with no-signaling certified deletion for \emph{any monotone access structure}, and (ii) a \emph{threshold} secret sharing scheme with adaptive certified deletion. Our first construction uses Bartusek and Khurana's (CRYPTO 2023) 2-out-of-2 secret sharing scheme with certified deletion as a building block, while our second construction is built from scratch and requires several new technical ideas. For example, we significantly generalize the ``XOR extractor'' of Agarwal, Bartusek, Khurana, and Kumar (EUROCRYPT 2023) in order to obtain better seedless extraction from certain quantum sources of entropy, and show how polynomial interpolation can double as a high-rate randomness extractor in our context of threshold sharing with certified deletion.
    \fi


\end{abstract}
\section{Introduction}

Secret sharing \cite{Shamir79,Blakley1899SafeguardingCK,ItoSaiNis87} is a foundational cryptographic primitive that allows a dealer to distribute a secret $s$ among 
$n$ parties so that only certain ``authorized'' subsets of the parties may recover the secret. A particularly common scenario is $(k,n)$ \emph{threshold} secret sharing, where the dealer splits $s$ into $n$ shares such that any $k$ of the shares can be combined to recover the secret $s$, but any $k-1$ or fewer shares leak no information about $s$. However, one can also consider a much more versatile setting, in which the authorized subsets of $n$ are defined by any \emph{monotone access structure} $\bbS$.\footnote{A set $\bbS$ of subsets of $[n]$ is monotone if for any subset $S \in \bbS$, it holds that $S' \in \bbS$ for all supersets $S' \supset S$.} Secret sharing schemes are ubiquitous in cryptography, and we refer the reader to Beimel's survey \cite{10.1007/978-3-642-20901-7_2} for a broader discussion\ifextendedabstract\else, including several applications.\fi

\ifextendedabstract
\else
A particularly appealing aspect of secret sharing that sets it apart from most other cryptographic primitives is that it \emph{doesn't require computational hardness assumptions}. That is, one can construct secret sharing schemes for arbitrary monotone access structures secure against any computationally unbounded adversary (e.g.\ \cite{ItoSaiNis87,C:BenLei88,10.1145/3188745.3188936}). 
\fi

\ifextendedabstract
\else
However, the security of these schemes still rests on a stringent assumption: over the course of the (potentially unbounded) adversary's operation, it only ever sees an unauthorized set of shares. This may be unacceptable for users sharing particularly sensitive information. Even if an adversary initially may only access a limited number of shares, over time they may be able to corrupt more and more parties, or perhaps more and more shares become compromised independently and are leaked into the public domain. A user who becomes paranoid about this possibility generally has no recourse, and, worse yet, cannot even \emph{detect} if an adversary has obtained access to enough shares to reconstruct their secret. 
\fi

In this work, we ask whether it is possible to strengthen the standard notion of secret sharing security, and relax the assumption that the adversary only ever corrupts an unauthorized set of parties. 
In particular:

\begin{quote}
\begin{center}
    \emph{Is it possible to achieve meaningful secret sharing security against adversaries that may eventually corrupt authorized sets of parties?}
\end{center}
\end{quote}

\ifextendedabstract
\else
Now, if the shares consist of only classical data, then there is no hope of answering the above question in the affirmative. Indeed, once an adversary obtains any share, they can make a copy and store it away. Once they've collected and stored an authorized set, they'll be able to recover the secret due to the correctness of the secret sharing scheme. 
\fi

\paragraph{Certified deletion.} On the other hand, the \emph{uncertainty principle} of quantum mechanics offers some hope: if shares are encoded into quantum states, then the useful share information may be erased by applying certain destructive measurements. Thus, a user that is worried about an adversary eventually corrupting an authorized set of shares may request and verify that this ``deletion'' operation is performed on some set of their shares. Now, even if the adversary learns enough shares in the future to constitute an authorized set, the already-deleted shares will remain useless, and there is hope that the user's secret remains private.

\ifextendedabstract
\else
Indeed, the basic idea of leveraging the uncertainty principle to perform ``certified deletion'' of private information was first put forth by Broadbent and Islam \cite{TCC:BroIsl20} in the context of one-time-pad encryption, and has since been applied in several contexts throughout cryptography, e.g.\
\cite{AC:HMNY21, C:HMNY22, poremba:LIPIcs.ITCS.2023.90, C:BarKhu23, 10.1007/978-3-031-48624-1_7, HKMNPY24, BGKMRR24}. In fact, Bartusek and Khurana \cite{C:BarKhu23} previously constructed a very limited flavor of secret sharing with certified deletion, namely, 2-out-of-2 secret sharing where only one of the two shares admits the possibility the deletion. 
\fi
\ifextendedabstract
\else
Their scheme allows a user to split a secret $s$ into a quantum share $\ket{\sh_1}$ and a classical share $\sh_2$. If an adversary first obtains and deletes $\ket{\sh_1}$, then obtains $\sh_2$, it will still be unable to reconstruct the secret $s$. One can also view the original one-time-pad encryption with certified deletion scheme of \cite{TCC:BroIsl20} as exactly this flavor of 2-out-of-2 secret sharing with certified deletion, where the quantum share is the ciphertext and the classical share is the secret key.
\fi
In this work, we show that it is possible to introduce certified deletion guarantees into more versatile and general flavors of secret sharing, addressing several definitional and technical issues along the way.


\subsection{Our Results}

We formulate two powerful but incomparable notions of certified deletion security for general-purpose secret sharing schemes, and show how to construct a scheme satisfying each definition. One of our key technical tools is a high-rate seedless extractor from certain quantum sources of entropy that significantly generalizes and improves upon the ``XOR extractor'' of \cite{EC:ABKK23}.

\paragraph{No-signaling security.}

First, we address the shortcomings of
\cite{C:BarKhu23}'s security definition for 2-out-of-2 secret sharing sketched above, and formulate a natural extension that (i) applies to schemes for \emph{any} monotone access structure, and (ii) allows for the possibility that \emph{any} of the shares may be deleted. 

\ifextendedabstract
\else
In particular, we model a scenario involving multiple non-communicating adversaries that each individually have access to some unauthorized set of shares. These adversaries may even share entanglement, but may not exchange messages. 
Now, the user may request that some of its shares are deleted. If the adversaries jointly delete enough shares so that the remaining undeleted shares form an \emph{unauthorized} set, then we combine the views of all the adversaries together, and require that the user's secret remains private even given this joint view.
That is, even if a single adversarial entity is eventually able to corrupt up to \emph{all} of the parties, they will still not be able to recover the secret if enough shares have previously been deleted.\footnote{We remark that this definition also captures adversaries that don't end up corrupting \emph{all} the shares, by imagining that there is a separate component of the adversary that honestly deletes the uncorrupted shares.}
\fi

We refer to this security notion for secret sharing schemes as \emph{no-signaling security} 
\ifextendedabstract\else(see \cref{sec:definitions} for a precise definition)\fi,
emphasizing the fact that shares must be deleted by adversaries that cannot yet pool information about an authorized set of shares, as this would trivially allow for reconstruction of the secret. Then, \ifextendedabstract\else in \cref{sec:no-signaling-construction} \fi we show how to combine \cite{C:BarKhu23}'s simple 2-out-of-2 secret sharing scheme with any standard secret sharing scheme for monotone access structure $\bbS$ (e.g.\ \cite{ItoSaiNis87,C:BenLei88,10.1145/3188745.3188936}) in order to obtain a secret sharing scheme for $\bbS$ with no-signaling security.

\begin{theorem}[Informal]
There exists a secret sharing scheme with no-signaling certified deletion security for any monotone access structure $\bbS$.
\end{theorem}

\paragraph{Adaptive security.} 
\ifextendedabstract
\else
Next, we consider a particularly cunning but natural class of adversaries that exhibit the following behavior. Suppose that initially the adversary only obtains access to some unauthorized set of shares. At some point, the user becomes paranoid and requests that some subset of these shares are deleted. The adversary obliges but then continues to corrupt new parties or locate other leaked shares. The adversary may continue to delete some of these shares to appease the user, while continuing to work behind the scenes to mount a long-term attack on the system. However, as long as the set of corrupted parties minus the set of certifiably deleted shares continues to be unauthorized, we can hope that the user's secret remains private from such an adversary.
\fi

\ifextendedabstract
\else
Unfortunately, the notion of no-signaling security does not capture the behavior of such an \emph{adaptive} adversary. That is, no-signaling security only models adversaries that delete once, and then receive some extra information after this single round of deletion. Thus, we formalize \emph{adaptive security} as an alternative and quite strong notion of certified deletion security for secret sharing schemes \ifextendedabstract\else(see \cref{sec:definitions} for a precise definition).\fi
\fi

Protecting against such arbitrarily adaptive adversaries turns out to be a significant challenge. The main technical component of our work realizes a secret sharing scheme with adaptive certified deletion security for the specific case of threshold access  
\ifextendedabstract\else structures (\cref{sec:adaptive-construction}).\fi

\begin{theorem}[Informal]
There exists a threshold secret sharing scheme with adaptive certified deletion security.
\end{theorem}

\paragraph{High-rate seedless extractors from quantum sources of entropy.} One of our technical building blocks is an improved method for seedless extraction from certain quantum sources of entropy. Roughly, the source of entropy comes from performing a standard basis measurement on a register that is in superposition over a limited number of Fourier basis states. 

While seedless extraction from such sources of entropy \cite{EC:ABKK23} has been a crucial component in previous realizations of cryptographic primitives with certified deletion \cite{C:BarKhu23},\footnote{See discussion \ifextendedabstract\else therein\fi for why \emph{seedless} as opposed to seeded extraction is crucial.} the technique had been limited to (i) extracting from qubit registers (i.e.\ where data is natively represented as superpositions of bitstrings) and (ii) extracting only a single bit of entropy. Here, we generalize these techniques to extract from qu\emph{dit} registers (i.e.\ where data is natively represented as superpositions of vectors over finite fields), and produce several field elements worth of entropy, vastly improving the rate of extraction. 
\ifextendedabstract
\else
Beyond being interesting in its own right, it turns out that these improvements are crucial for showing security our construction of threshold sharing with adaptive certified deletion. 
Moreover, we show how these high-rate extraction techniques can be applied to extension fields, meaning that we can represent our quantum shares as string of qubits (as opposed to qudits), removing the need for entanglement in our construction. 
\ifextendedabstract We refer the reader to the full version for more details.\else We refer the reader to \cref{sec:overview-extractor} and \cref{sec:extractor} for more details.\fi
\fi

\ifextendedabstract
\else
\section{Technical Overview}

Intuitively, certified deletion for secret sharing aims to keep the secret private from an adversary if the total set of undeleted shares they have access to is unauthorized. One could formalize this by considering an adversary who initially receives an unauthorized set of shares and then deletes some of them. If the undeleted shares are still unauthorized when combined with the shares that the adversary did not receive, then we allow the adversary to access these remaining shares. This closely matches the definition of encryption with certified deletion, where the adversary initially receives and deletes a ciphertext $\mathsf{Enc}(k, m)$ encrypting message $m$ using key $k$, and then later receives the key $k$.

However, this definition is not meaningful for all access structures. For example, in a $k$ out of $n$ access structure where $k < n/2$, the shares that the adversary does not start with \emph{already} form an authorized set on their own, so it never makes sense to allow the adversary to access all of these shares at once. In this section, we give an overview of two different ways to address this definitional deficiency: no-signaling certified deletion and adaptive certified deletion.

\subsection{No-Signaling Certified Deletion}

In no-signaling certified deletion, we address this problem by allowing the adversary to delete from multiple sets of shares $P_1, \dots, P_\ell$. However, if $P_1\cup \dots \cup P_\ell$ contains all shares, then the adversary as a whole gets to see every share before it generates any deletion certificates. Thus, to prevent trivial attacks, we do not allow the adversary to communicate across sets. However, the different parts of the adversary may still \emph{share entanglement}. This modification yields the no-signaling certified deletion game $\NSCD_{\bbS}(s)$ for secret $s$ and access structure $\bbS$ over $n$ parties, which we describe here.
\begin{enumerate}
    \item The challenger secret-splits $s$ into $n$ shares with access structure $\bbS$.
    \item Each adversary $\Adv_i$ is initialized with one register of a shared state $\ket{\psi}$,  receives the shares in a set $P_i$, and produces some set of certificates $\{\cert_j\}_{j \in P_i}$. If $\Adv_i$ does not wish to delete share $j$, then it may set $\cert_j = \bot$.
    \item If the total set of shares that have not been deleted is unauthorized, then output the joint view of the adversaries. Otherwise, output $\bot$.
\end{enumerate}

\noindent No-signaling certified deletion for secret sharing requires that for every secret pair $(s_0, s_1)$ and every partition $P = (P_1, \dots, P_\ell)$ of $[n]$, the outputs of $\NSCD_{\bbS}(s_0)$ and $\NSCD_{\bbS}(s_1)$ have negligible trace distance.

\paragraph{Tool: 2-of-2 Secret Sharing with Certified Deletion~\cite{C:BarKhu23}.} Recently, Bartusek and Khurana constructed a variety of primitives with certified deletion. One of these primitives is a secret sharing scheme which splits a secret $s$ into a quantum share $\ket{\sh_1}$ and a classical share $\sh_2$, along with a verification key $\verkey$ that can be used to test the validity of deletion certificates. Given either one of the shares, the secret is hidden. Furthermore, if an adversary given $\ket{\sh_1}$ performs a destructive measurement that yields a valid deletion certificate, then they will never be able to recover $s$, even if they later obtain $\sh_2$. Note that in this scheme, only one of the two shares can be deleted.

\paragraph{A Black-Box Compiler.} We show how to compile Bartusek and Khurana's 2-of-2 certified deletion scheme together with any classical secret sharing scheme into a secret sharing scheme with no-signaling certified deletion. Notably, the compiled scheme inherits the same access structure as the classical secret sharing scheme. Thus, one can construct secret sharing with no-signaling certified deletion for general access structures by using any classical secret sharing scheme for general access structures, e.g.~\cite{ItoSaiNis87,EC:ABFNP19}. 

As a starting point, let us first construct a scheme where only one of the shares can be deleted.
\begin{enumerate}
    \item Secret split the secret $s$ into a quantum share $\ket{\qsh}$ and a classical share $\csh$ using the 2-of-2 secret sharing scheme with certified deletion. This also produces a verification key $\vk$.
    \item Split the 2-of-2 classical share $\csh$ into classical shares $\csh_{1},\dots, \csh_{n}$ using the classical secret sharing scheme for $\bbS$.
    \item The verification key is $\verkey$ and the $i$'th share is $\csh_i$. The deletable quantum share is $\ket{\qsh}$. 
\end{enumerate}
Given the quantum share and any authorized set of classical shares, $s$ can be reconstructed by first recovering $\csh$ from the authorized set.
On the other hand, any adversary which attempts to delete $\ket{\qsh}$ with only access to an unauthorized set of classical shares has no information about the 2-of-2 classical share $\csh$. Thus if they produce a valid deletion certificate, they will have no information about $s$ even after obtaining the rest of the classical shares, which only reveals $\csh$.

\paragraph{Who Deletes?} To finish the compiler, we need to enable certified deletion of \emph{any share}. 
This can be achieved by adding a step at the beginning of the compiler to create $n$ classical shares $\sh_1, \dots, \sh_n$ of $s$ with the same access structure $\bbS$. Then, the splitter can enable certified deletion for each share $\sh_i$ by using the prior compiler to produce a set of classical shares $\csh_{i, 1}, \dots, \csh_{i,n}$, a deletable quantum share $\ket{\qsh_i}$, and a verification key $\verkey_i$. The $i$'th share contains the deletable state $\ket{\qsh_i}$, as well as $\{\csh_{j,i}\}_{j \in [n]}$.

Note that anyone holding share $i$ is able to delete $\sh_i$ by deleting $\ket{\qsh_i}$, as discussed previously. If sufficiently many shares are deleted, so that the only remaining $\sh_i$ form an unauthorized set, then no adversary can learn anything about the secret even after obtaining all of the remaining shares and residual states.

\paragraph{Proof of Security: Guessing Deletions.} Although the intuition is straightforward, there is a nuance in the proof of security. When proving security, we wish to replace the deleted 2-of-2 secrets $\sh_i$ with empty secrets $\bot$. If we could do so, then security immediately reduces to the security of the classical $\bbS$-scheme, since only an unauthorized set of $\sh_i$ remains. However, it is difficult to determine which of these 2-of-2 secrets $\sh_i$ to replace with $\bot$ when preparing the shares. 

Since non-local operations commute, we could consider generating the shares for each adversary $\Adv_i$ one at a time. For example, supposing $\Adv_1$ operates on the set of shares $P_1 \subset [n]$, the experiment could initialize $\Adv_1$ with uniformly random shares, and then for each $i \in P_1$, reverse-sample the shares $\{\csh_{i,j}\}_{j\in[n] \setminus P_1}$ for the rest of the adversaries to match either $\sh_i$ or $\bot$, depending on whether or not $\Adv_1$ deleted share $i$.

Unfortunately, we cannot continue this strategy for all of the adversaries. It may be the case that the union of $\Adv_1$ and $\Adv_2$'s shares $P_1 \cup P_2$ contains an authorized set. Thus, when initializing $\Adv_2$, the challenger must already know whether, for each $i \in P_2$, the $i$'th share of $s$ should be set to $\sh_i$ or $\bot$ (since this will be determined by $\{\csh_{i,j}\}_{j \in P_1 \cup P_2}$). This view is constructed before the adversary decides whether or not to delete share $i$, so the only way for the challenger to do this is to guess whether the adversary will delete share $i$ or not. 

Now, guessing which shares the entire set of adversaries will delete incurs a multiplicative exponential (in $n$) loss in security. Fortunately, Bartusek and Khurana's 2-of-2 scheme actually satisfies an \emph{inverse exponential} trace distance between the adversary's view of any two secrets, after deletion. Thus, by setting the parameters carefully, we can tolerate this exponential loss from guessing, and show that our scheme for general access structures satisfies negligible security. 



\subsection{Adaptive Certified Deletion}\label{sec:overview-adaptive-cd}

Intuitively, any definition of certified deletion should allow the adversary to eventually receive an authorized set of shares, as long as they have previously deleted enough shares so that their total set of undeleted shares remains unauthorized. In no-signaling certified deletion, we allowed multiple non-communicating adversaries to delete from different unauthorized sets of shares. That is, when $\Adv_i$ generates its set of certificates, it only has access to a single unauthorized set $P_i$. However, one could also imagine a more demanding setting where, after deleting some shares, the adversary can adaptively corrupt \emph{new shares}, as long as their total set of undeleted shares remains unauthorized. Then, they can continue deleting shares and corrupting new shares as long as this invariant holds. This setting arises naturally when we consider an adversary which covertly compromises shares over a long period of time, while occasionally deleting shares to avoid revealing the extent of the infiltration. We call this notion \emph{adaptive certified deletion}. It is defined using the following adaptive certified deletion game $\ACD_{\bbS}(s)$.
\begin{enumerate}
    \item The challenger splits the secret $s$ into $n$ shares with access structure $\bbS$. The adversary starts with an empty corruption set $C$ and an empty deletion set $D$.
    
    \item For as many rounds as the adversary likes, it gets to see the shares in $C$ and choose whether to corrupt or delete a new share.

        \textbf{Corrupt a new share.} The adversary corrupts a new share $i$ by adding $i$ to $C$. If $C\backslash D$ is authorized, the experiment immediately outputs $\bot$.
        
        \textbf{Delete a share.} The adversary outputs a certificate $\cert$ for a share $i$. If $\cert$ is valid, add $i$ to $D$. Otherwise, the experiment immediately outputs $\bot$.
        
    \item When the adversary is done, the experiment outputs its view.
\end{enumerate}

Adaptive certified deletion for secret sharing requires that for every secret pair $(s_0, s_1)$, the outputs of $\ACD_{\bbS}(s_0)$ and $\ACD_{\bbS}(s_1)$ have negligible trace distance. In this work, we focus on the $(k,n)$ threshold access structure, where any set of size $\geq k$ is authorized. 

\paragraph{Incomparable Definitions.}
We have already seen that no-signaling certified deletion does not imply adaptive certified deletion. It is also the case that adaptive certified deletion does not imply no-signaling certified deletion. Consider a two-part no-signaling adversary $\Adv_1$ and $\Adv_2$ with views $P_1$ and $P_2$. To change $(\Adv_1, \Adv_2)$ to an adaptive adversary, one would need to come up with a transformation that deletes the same shares as $(\Adv_1, \Adv_2)$, in the same way. However, $\Adv_1$ might not even decide which shares to delete until after they have seen every share in $P_1$. So, the new adaptive adversary would have to corrupt all of $P_1$ before it can delete a single share that $\Adv_1$ would. Similarly, it would also have to corrupt all of $P_2$ before it knows which shares $\Adv_2$ would delete. However, if $P_1\cup P_2$ is authorized, then the experiment would abort before the new adaptive adversary gets the chance to delete shares for both $\Adv_1$ and $\Adv_2$.

\paragraph{An Attack on the No-Signaling Construction.} Unfortunately, the previous construction actually does \emph{not} in general satisfy adaptive certified deletion security. Indeed, observe that the classical parts of each share can never be deleted. Because of this, an adversary could, for any $i$, obtain $k$ classical shares $\csh_{i,1},\dots, \csh_{i,k}$ that reveal $\csh_i$, simply by corrupting and immediately deleting the first $k$ shares one-by-one. Afterwards, the adversary will always have \emph{both} the classical share $\csh_i$ and the quantum share $\ket{\qsh_i}$ when it corrupts a new share $i$, so it can recover the underlying classical share $\sh_i$. Now it can ``delete'' the $i$'th share while keeping $\sh_i$ in its memory. Eventually, it can collect enough $\sh_i$ in order to obtain the secret $s$.


The core problem with the no-signaling construction is the fact that it encodes the 2-of-2 classical shares $\csh$ in a form which can never be deleted.
If we were to take a closer look at Bartusek and Khurana's 2-of-2 scheme, we would observe that $\csh$ contains a mapping $\theta$ of which parts of $\ket{\qsh}$ encode the secret (the ``data indices'') and which parts contain only dummy information used to verify certificates (the ``check indices''). Without $\theta$, there is no way to decode the secret.
Unfortunately, in order to encode $\theta$ in a deletable form, we seem to be back where we started - we need secret sharing with adaptive certified deletion!



\paragraph{A New Construction.} To avoid this pitfall, we take a new approach that allows the parties to reconstruct \emph{without knowledge of the check indices}. This removes the need to encode $\theta$ altogether. To achieve this, we begin with Shamir's secret sharing \cite{Shamir79}, in which the shares are evaluation points of a degree $k-1$ polynomial $f$ where $f(0) = s$. This polynomial is over some finite field $\bbK$ with at least $n+1$ elements. A useful property of Shamir's secret sharing is that it has good error-correcting properties - in fact, it also forms a Reed-Solomon code, which has the maximum possible error correction~\cite{ReedSolomon60}. 

To split a secret $s$, we start by constructing a polynomial $f$ where $f(0) = s$. Each share contains some number of evaluations of $f$ encoded in the computational basis. These evaluations are mixed with a small number of random Fourier basis states that will aid in verifying deletion certificates. The positions of these checks, along with the value encoded, make up the verification key. An example share and key are illustrated here.

\[
\begin{array}{cccccccccccc}
     \sh =  &\ket{f(1)} &\otimes& \ket{f(2)} &\otimes&  \mathsf{QFT}\ket{r_1} &\otimes&  \ket{f(4)} &\otimes& \mathsf{QFT}\ket{r_2} &\otimes&  \dots 
     \\
     \verkey =  &* && * && r_1 && * && r_2 && \dots
\end{array}
\]

\noindent When reconstructing the secret, these checks are essentially random errors in the polynomial evaluation. By carefully tuning the degree of the polynomial together with the number of evaluations and checks in each share, we can ensure that any $k$ shares contain enough evaluation points to correct the errors from the check positions, but that any $k-1$ shares do not contain enough evaluation points to determine the polynomial. This results in $f$ being determined by slightly more than $k-1$ shares worth of evaluations. Additionally, we slightly increase the degree of the polynomial to account for the limited amount of information that an adversary can retain after deletion. See \Cref{sec:adaptive-construction} for more details.

Share deletion and certificate verification follow the established formula. To delete a share, measure it in the Fourier basis and output the result as the certificate. To verify the certificate, check that it matches the verification key at the check positions.

\paragraph{Proving Adaptive Certified Deletion.}\label{sec:overview-proving-adaptive}
Intuitively, we want to show that after the adversary deletes a share, the next share it corrupts gives it no additional information, no matter how many shares the adversary has seen so far. To formalize this, we will show that the adversary cannot distinguish between the real $\ACD_{(k,n)}(s)$ experiment and an experiment in which each share is generated uniformly at random and independently of the others. Since the first $k-1$ shares to be corrupted do not yet uniquely determine the polynomial $f$, they already satisfy this. Thus, we can restrict our attention to modifying the last $n-k+1$ shares to be corrupted.

It will be useful to name the shares in the order they are corrupted or deleted. The $\corctr$'th share to be corrupted is $c_\corctr$, and the $\delctr$'th share to be deleted is share $d_\delctr$. In the $(k,n)$ threshold case, if $c_{k-1+\delctr}$ is corrupted before $d_{\delctr}$ is deleted, then $C\backslash D$ has size $k$ and is authorized, so the experiment will abort.

\paragraph{Techniques from BK23.} We begin by recalling the techniques introduced in \cite{C:BarKhu23} to analyze 2-of-2 secret sharing with certified deletion, along with the construction. These techniques will form the starting point of our proof. To share a single-bit secret $s \in \{0,1\}$, sample random $x, \theta \gets \{0,1\}^\secpar$ and output 
\[
    \sh_1 = H^{\theta} \ket{x},\quad
    \sh_2 = \left(\theta, s \oplus \bigoplus_{i:\theta_i = 0} x_i\right),\quad
    \vk = (x, \theta),
\]

\noindent where $H^\theta$ denotes applying the Hadamard gate $H$ to the $i$'th register for each $i : \theta_i = 1$. Bartusek and Khurana showed that if an adversary given $\sh_1$ produces a certificate $\cert$ such that $\cert_i = x_i$ for every check position $i : \theta_i = 1$, then they cannot distinguish whether $s=0$ or $s=1$ even if they later receive $\sh_2$. Their approach has three main steps. 
\begin{enumerate}
    \item First, they delay the dependence of the experiment on $s$ by initializing $\sh_1$ to be the register $\calX$ in $\sum_{x}\ket{x}^{\calX} \otimes \ket{x}^{\calY}$. Later, the challenger can obtain $x$ by measuring register $\calY$, and use it to derive $s$. 
    
    
    \item Second, they argue that if the adversary produces a valid deletion certificate, then $\sh_1$ has been ``almost entirely deleted'', in the sense that the challenger's copy satisfies a checkable predicate with high probability. Intuitively, this predicate shows that the data positions ($\theta_i = 0$) of the challenger's copy have high joint entropy when measured in the computational basis. To show that the predicate holds, they use the fact that the adversary does not have $\sh_2$ in their view, so $\sh_1$ looks uniformly random. This allows a cut-and-choose argument where the locations of the check indices are determined \emph{after} the adversary outputs its deletion certificate.
    
    \item Finally, they show that the challenger derives a bit $s$ that is uniformly random and independent of the adversary's view. This utilizes a result from \cite{EC:ABKK23} which shows that XOR is a good seedless extractor for entropy sources that satisfy the aforementioned predicate.
\end{enumerate}

\paragraph{Adapting to Secret Sharing.} As a starting point, let us try to adapt these techniques to undetectably change shares to uniformly random. For concreteness, consider the task of switching a share $c_{k-1+\delctr}$ to uniformly random. Although we have not yet outlined the general proof structure, we will eventually need to perform this task. We will use this starting point to gain insights that will help guide the eventual proof structure.

\begin{enumerate}
    \item The first step is to delay the synthesis of the secret information until after the adversary outputs a deletion certificate. In our case, we will delay creating share $c_{k-1+\delctr}$ until after the adversary produces a valid certificate for share $d_{\delctr}$.

    This can be achieved by sampling the first $k-1$ corrupted shares to be uniformly random, then using polynomial interpolation to prepare the rest of the shares. 
    More concretely, consider the first corrupted $k-1$ shares $c_1, \dots, c_{k-1}$. The challenger will prepare each of these shares $c_\corctr$ using two registers $\chalReg_{c_\corctr}$ and $\calS_{c_\corctr}$, then send the share to the adversary in register $\calS_{c_\corctr}$.
    To prepare the $j$'th position of $c_\corctr$, the experiment challenger prepares either a uniform superposition $\sum_{x \in \bbK}\ket{x}^{\calC_{c_\corctr,j}}$ or $\sum_{x \in \bbK}\QFT\ket{x}^{\calC_{c_\corctr,j}}$, depending on whether $j$ is an evaluation position or a check position. If $j$ is an evaluation position for share $c_{\corctr}$, the experiment challenger copies $\calC_{c_{\corctr},j}$ to $\calS_{c_{\corctr},j}$ in the computational basis, yielding \[\propto \sum_{x \in \bbK}\ket{x}^{\calS_{c_\corctr,j}} \otimes \ket{x}^{\calC_{c_\corctr,j}},\]
    and otherwise it copies $\calC_{c_{\corctr},j}$ to $\calS_{c_{\corctr},j}$ in the Fourier basis, yielding \[\propto \sum_{x \in \bbK}\QFT\ket{x}^{\calS_{c_\corctr,j}} \otimes \QFT\ket{x}^{\calC_{c_\corctr,j}}.\] 

    Note that the adversary cannot determine which positions are evaluation positions and which are check positions, since each $\calS_{c_a,j}$ register looks maximally mixed. Also observe that $\chalReg_{c_\corctr}$ contains a copy of share $c_\corctr$, and the evaluation positions in the initial $k-1$ $\chalReg_{c_\corctr}$ registers determine the polynomial $f$. Then, when the adversary requests to corrupt a later share, the challenger computes the evaluation points for that share by performing a polynomial interpolation using its copies of the prior shares. For reasons that will become apparent shortly, we require that share $d_\delctr$ is included when interpolating $c_{k-1+\delctr}$. The other points may be arbitrary.

    The above procedure is actually slightly simplified; since the degree of $f$ is slightly larger than the number of evaluation positions in $k-1$ shares, the first $k-1$ shares do not quite determine $f$. To remedy this, we will also initialize a small portion of \emph{every} $\calS_i$ to be uniformly random, before any interpolation takes place.

    \item Next, we will need to show that $\chalReg_{d_{\delctr}}$, which contains the challenger's copy of share $d_{\delctr}$, satisfies the deletion predicate if $\cert_{d_{\delctr}}$ passes verification.
    This is not hard to show if $d_\delctr$ was generated uniformly at random, but it is not clear what happens if the adversary has some information about where the check positions are in $d_{\delctr}$ before deleting it.
    The first $k-1$ shares are uniformly random in the original experiment, so this is not a problem for any share which is deleted before $c_k$ is corrupted. However, later shares depend on earlier shares, potentially leaking information about the check positions. This will be our first barrier to overcome.
    

    \item Finally, we will need to show that interpolating $c_{k-1+\delctr}$ using $\chalReg_{d_{\delctr}}$ produces a uniformly random value whenever $\chalReg_{d_{\delctr}}$ satisfies the deletion predicate. In other words, \textbf{polynomial interpolation should double as a good randomness extractor from deleted shares}. Fortunately, polynomial interpolation is a matrix multiplication, and we have intuition from the classical setting that linear operations are good randomness extractors.
    Since a small amount of every share is uniformly random ``for free'', the extractor needs to produce only slightly less than a full share's worth of evaluations to produce $c_{k-1+\delctr}$. This is our second technical contribution, which we will revisit in \Cref{sec:overview-extractor}.
\end{enumerate}

In step two, we seem to need the evaluation points in $d_{\delctr}$ to look like the check positions when $d_{\delctr}$ is deleted, i.e. they should be uniformly random and independent of the rest of the adversary's view.
A natural approach to ensure this is to modify the shares to uniformly random round-by-round over a series of hybrid experiments. 
In hybrid $i$, the first $k-1+i$ shares to be corrupted are uniformly random. Since $d_{i+1}$ must be deleted before $c_{k+i}$ is corrupted (or else the experiment aborts), $d_{i+1}$ must have been one of the uniformly random shares. 
Now we can apply the cut-and-choose argument to show that $\chalReg_{d_i}$ satisfies the deletion predicate, thereby satisfying the extractor requirements to change $c_{k+i}$ to be uniformly random and reach hybrid $i+1$.
The first $k-1$ shares are already uniformly random, which gives us an opening to begin making modifications in round $k$.

Unfortunately, the strategy of modifying the shares one-by-one to be uniformly random has a major flaw.  In particular, the challenger needs to produce additional polynomial evaluations whenever the adversary wishes to corrupt another share, which it does via interpolation. Recall that in order to claim that $c_{k-1+i}$ is indistinguishable from random, we apply an extractor which uses register $\chalReg_{d_i}$ as its source of entropy. 
But in order to invoke the security of the extractor, it seems that we cannot allow the challenger to re-use $\chalReg_{d_i}$ when interpolating later points, as this might leak additional information about the source to the adversary. 

To get around this issue, we might require that the challenger never uses $\calC_{d_i}$ again to interpolate later points. However, the randomness extractor outputs \emph{less} randomness than the size of a share. Intuitively, this occurs because the adversary can avoid fully deleting the source share $d_{i}$ by guessing the location of a very small number of check positions.\footnote{One may wonder whether it is possible to instead use coset states for the shares, which provide guarantees of \emph{full} deletion \cite{BGKMRR24}. Unfortunately, coset states induce errors which are the sum of a small number of uniformly random vectors. It is not clear how to correct these errors to reconstruct the secret without prior knowledge of the underlying subspace. However, encoding the subspace brings us back to the original problem of secret sharing with adaptive certified deletion.}
Imperfect deletion limits the entropy of the source, which in turn limits the size of the extractor output. 
Now, since the challenger started with \emph{exactly} enough evaluations to uniquely determine $f$, if we take away the points in $\chalReg_{d_i}$ then there are no longer enough evaluation points remaining to create the rest of the shares, even given the newly interpolated points.  


\paragraph{Predicates First, Replacement Later.} Intuitively, the problem outlined above arises from the possibility that the adversary receives additional information about earlier shares from the later ones, since they are all correlated through the definition of the polynomial $f$. Our first idea to overcome this issue is to prove that the predicate holds for all rounds \emph{before} switching any shares to uniformly random. In particular, we will consider a sequence of hybrid experiments where in the $i$'th hybrid, the challenger performs the predicate measurement on $\chalReg_{d_i}$ after receiving and verifying the corresponding certificate. If the measurement rejects, the experiment immediately aborts and outputs $\bot$. 

If we can undetectably reach the last hybrid experiment, then it is possible to undetectably replace every share with uniform randomness by working backwards. In the last hybrid experiment, either the predicate holds on the challenger's copy $\chalReg_{d_{n-k+1}}$ of share $d_{n-k+1}$ or the experiment aborts. In either case, the last share $c_n$ to be corrupted can be undetectably replaced with uniform randomness. Since no further shares are interpolated, we no longer run into the issue of re-using the randomness source $\chalReg_{d_{n-k+1}}$, allowing the challenger to safely complete the experiment. Then, once $c_n$ is uniformly random, the challenger no longer needs to interpolate shares after $c_{n-1}$, so $c_{n-1}$ can also be replaced with uniform randomness. This argument can be continued until all shares are replaced.

To undetectably transition from hybrid $i$ to hybrid $i+1$, we must show that the predicate measurement returns success with high probability on $\chalReg_{d_{i+1}}$. This is not hard to show for shares which are deleted \emph{before} the $k$'th share is corrupted, because the deleted shares must be one of the shares which were generated uniformly at random. However, it is not clear how to show for shares which are deleted \emph{after} the $k$'th share is corrupted, since this seems to require replacing $c_k$ with uniform randomness, which brings us back to our previous problem.

\paragraph{Chaining Deletions via Truncated Experiments.} Our second insight is the observation that the result of a measurement made when $d_i$ is deleted \emph{is independent of later operations}. Thus, when arguing about the probability that the predicate measurement accepts on $\chalReg_{d_i}$, it is sufficient to argue about the truncated experiment that ends immediately after the predicate measurement on $\chalReg_{d_i}$.
Crucially, the adversary cannot corrupt share $c_{k+i}$ in the truncated experiment without causing the experiment to abort due to $|C\backslash D| \geq k$. Instead, $c_{k-1+i}$ is the last share that can be corrupted.
This prevents the catastrophic re-use of $\chalReg_{d_i}$ after share $c_{k-1+i}$ is constructed.

Let us assume that we have already shown that the deletion predicate measurement accepts on $\chalReg_{d_i}$ with high probability; for example, this clearly holds for $d_1$, which must be corrupted before $c_k$ is corrupted. How likely it is to accept on $\chalReg_{d_{i+1}}$?
Say the deletion predicate measurement accepts on $\chalReg_{d_i}$. Then we can invoke the extractor to undetectably replace share $c_{k-1+i}$ with uniform randomness in the truncated game, since no further shares are corrupted before the game ends.
We can use similar logic to replace each of the first $k-1+i$ shares to be corrupted in the truncated game.
At this point, the adversary has no choice but to delete a uniformly random share as $d_{i+1}$, so we can apply a cut-and-choose argument to show that the predicate holds with high probability on $\chalReg_{d_{i+1}}$
This argument can be repeated inductively to show that the predicate holds in each of the polynomially many rounds.


\paragraph{Recap of the First Challenge.} In summary, the first challenge to address in proving adaptive certified deletion is the possibility of later shares leaking information about prior shares through the re-use of $\chalReg_{d_\delctr}$ in interpolation. This prevents directly replacing each share with uniform randomness. 
To sidestep this issue, we first argue that every $\chalReg_{d_\delctr}$ is a good source of entropy using a series of games which end after $d_\delctr$ is deleted. Then even if $\chalReg_{d_\delctr}$ is used to interpolate both share $c_{k-1+\delctr}$ and $c_{k+\delctr}$, we can rely on the entropy from $\chalReg_{d_{\delctr+1}}$ to mask its re-usage when interpolating $c_{k+\delctr}$.

\subsection{High Rate Seedless Extractors from Quantum Sources of Entropy}\label{sec:overview-extractor}

The final task to finish the proof of adaptive certified deletion security in the previous section is to show that polynomial interpolation is a good randomness extractor for entropy sources formed by deleted shares. Although polynomial interpolation arises quite naturally in our construction, there are additional technical reasons why it would be difficult to design a construction for adaptive certified deletion using existing extractors.

If we were to design a scheme using a \emph{seeded} extractor, as done in \cite{TCC:BroIsl20}, then every deletion would need to be independent of the seed to avoid the entropy source depending on the seed. However, as we saw with the no-signaling construction, safely encoding the seed seems to already require secret sharing with adaptive certified deletion.
\cite{C:BarKhu23} makes use of the seedless XOR extractor developed by \cite{EC:ABKK23} to avoid a similar problem. Unfortunately, the XOR extractor produces only a single bit from a relatively large input. In the case of threshold secret sharing, the extractor must use the randomness produced by deleting a single share to extract an output which is only slightly smaller than a full share.

To address this need, we give a new family of seedless randomness extractors for quantum entropy sources with high rate. These constructions have connections to linear error-correcting codes and may be of independent interest. 

\paragraph{A Family of Extractors.} 
The input of the extractor is a vector of $M$ elements of a finite field $\bbF$, and the output is a vector of $m$ elements of $\bbF$. The source consists of a register $\calX$ which may be arbitrarily entangled with a side-information register $\calA$. If the register $\calX$ is in superposition over Fourier basis vectors with Hamming weight $\leq (M-m)/2$ in $\bbF$, then we can argue that the output $\mathsf{Extract}(\calX)$ is uniformly random, even given $\calA$.\footnote{The Hamming weight over a (potentially non-binary) finite field is being the number of nonzero entries in the vector.}

The extractor family consists of matrices $R \in \bbF^{m \times M}$ such that every set of $m$ columns of $R$ are linearly independent. In other words, $R$ is a parity check matrix for a linear error-correcting code with distance at least $m$. 
An extractor $R$ is applied by coherently multiplying $R$ with the source register $\calX$ in the computational basis and writing the result to the output register.

\paragraph{Application to Polynomial Interpolation.} This family generalizes both the XOR extractor and polynomial interpolation. The XOR extractor can be represented as the all-ones matrix with one row. Each column is non-zero, so the extractor can produce a one-bit output.
In the case of polynomial interpolation, we can write the linear interpolation operator for a polynomial $f$ as a matrix $R$ with $\mathsf{deg}(f)+1$ columns and a number of rows equal to the number of points being interpolated. $R$ is a sub-matrix of a parity check matrix for a Reed-Solomon code, so it falls into the new extractor family. In fact, our result shows that any subset of columns in a polynomial interpolation matrix forms a good randomness extractor for an appropriate randomness source. When interpolating a share $c_{k-1+\delctr}$, we can write the interpolation matrix as $R = [R_1 | R_2]$, where $R_2$ is applied to $d_\delctr$ and $R_1$ is applied to the other points $x$ on the polynomial. Then the new share is $c_{k-1+\delctr} = R_1 x + R_2 d_{\delctr}$. If $d_{\delctr}$ has satisfies the deletion predicate, then our extractor result shows that $R_2 d_{\delctr}$ is uniformly random. Thus, the newly interpolated share $c_{\delctr}$ is also uniformly random.

\paragraph{Removing Entanglement.} A downside of using polynomials for secret sharing is that each evaluation point exists in field $\bbF$ whose size scales with the total number of distinct points that must be defined on the polynomial. For example, $\bbF$ might be $\bbZ_p$ for some prime $p> nt$, where $t$ is the number of evaluations per share.
Using the approach outlined so far, each check position must be encoded in the Fourier basis over the same field $\bbK$. However, a logical $\bbZ_p$-qudit requires $\lceil \log_2(nt+1)\rceil$ qubits, which must all be entangled to produce a Fourier basis element of $\bbZ_p$.

We show how to remove the entanglement of the construction to only use \textit{single-qubit} states, either in the Hadamard basis or in the computational basis. We modify the construction by setting the field $\bbF$ to be the binary extension field with $2^{\lceil\log_2(nt+1)\rceil}$ elements, so that each check position consists of $\lceil\log_2(nt+1)\rceil$ qubits. Then, we \emph{individually} set each of these qubits to be a random Hadamard basis element. The other parts of the construction remain the same. Note that computational basis vectors over $\bbF$ can be encoded as a tuple of computational basis qubits.

Proving the security of this modification requires an expansion of the extractor theorem to allow general finite fields $\bbF$, which may have $p^k$ elements for some prime $p$ and $k\geq 1$. A Fourier basis element for such a field is obtained by applying the quantum Fourier transform over the additive group of $\bbF$, which is $\bbZ_p^k$. In particular, a Fourier basis element of $\bbF$ consists of $k$ Fourier basis elements of $\bbZ_p$. In the case where $p=2$, these are single-qubit Hadamard basis elements.

We emphasize that the \emph{only} change is to modify how Fourier basis elements are defined by allowing general finite fields; both the extractor family and the Hamming weight requirement remain the same (i.e.\ they are still defined with respect to $\bbF$, \emph{not} $\bbZ_p$). 
To gain intuition about the usefulness of this statement, let us consider its application in our secret-sharing construction. Ideally, an honest deleter would measure each qubit of its share in the Hadamard basis. However, since the dealer can only verify check positions, which each consist of $\lceil\log_2(nt+1)\rceil$ qubits, we can only prove a bound on the Hamming weight of the deleted state over $\lceil\log_2(nt+1)\rceil$-sized chunks, which corresponds to $\bbF$. This matches the entropy source requirements of the theorem. On the other hand, the polynomial that the secret is encoded in is also over $\bbF$, so polynomial interpolation must take place over $\bbF$. This matches the extractor family.

\subsection{Open Problems}

Although our results significantly strengthen secret sharing to resist new classes of attacks, we have only scratched the surface of an area with many fascinating open problems. We mention a few of them here.

\begin{itemize}

    \item \textbf{Adaptive Certified Deletion for General Access Structures.} Against adaptive attacks, we construct a secret sharing scheme for the special case of threshold access structures. Is it possible to construct one for \emph{general access structures}?

    \item \textbf{Stronger Definitions.} We prove the security of our schemes against \emph{either} ``distributed'' attacks (i.e.\ no-signaling security) \emph{or} adaptive attacks. Can we (i) formulate natural security definitions that capture both types of attacks, and (ii) prove the security of secret sharing schemes under such all-encompassing definitions?  

    \item \textbf{Public Verification.} The question of publicly verifiable certificates for encryption with certified deletion has seen significant progress recently~\cite{AC:HMNY21,poremba:LIPIcs.ITCS.2023.90,C:BarKhuPor23,BGKMRR24,TCC:KitagawaNY23}. However, the techniques used seem to require the use of a classical secret to decode the plaintext. For secret sharing with certified deletion, this secret would need to also be encoded in a manner that can be certifiably deleted, as mentioned in \Cref{sec:overview-adaptive-cd}. Is it possible to construct secret sharing with \emph{publicly verifiable} certificates of deletion?

    \item \textbf{Other Threshold Primitives.} Aside from secret sharing, there are many other primitives which use thresholds or other access structures. For example, in threshold signatures, any $k$ members may non-interactively sign messages under a secret key split between $n$ parties~\cite{C:DesFra89}. Is it possible to construct \emph{threshold signatures or other threshold primitives} with certified deletion?

    \item \textbf{High Rate Commitments with Certified Deletion.} A commitment with certified deletion allows the committed message to be certifiably and information-theoretically deleted~\cite{C:HMNY22,C:BarKhu23}. However, current approaches either work in the random oracle model or require $\Theta(\secpar)$ qubits to commit to a single bit.
    Our new high-rate extractor (\Cref{lem:extractor}) provides a promising start to reduce the commitment overhead. Unfortunately, the proof technique pioneered by~\cite{C:BarKhu23} for the plain model requires guessing the committed message, which incurs a security loss that is exponential in the size of the message.
    Is it possible to overcome this difficulty and construct commitments with certified deletion that have are \emph{not much larger than the committed message}?
\end{itemize}
\ifsubmission \else
\section{Preliminaries}

\subsection{Quantum Computation}

For any set $S$, an $S$-qudit is a quantum state in the Hilbert space spanned by $\{\ket{s}: s\in S\}$. A quantum register $\calX$ contains some number of qubits. $\ket{x}^{\calX}$ denotes a quantum state $\ket{x}$ stored in register $\calX$. A classical operation $f$ can be applied to a quantum register $\calX$ using the map 
\[
    \ket{x}^{\calX} \otimes \ket{y}^{\calY} \mapsto \ket{x}^{\calX} \otimes \ket{y + f(x)}^{\calY}
\]
We denote the application of this operation as $\calY \gets f(\calX)$.

\begin{lemma}[Gentle Measurement \cite{TIF:Win99}]\label{lem:gentle-measurement}
Let $\rho^X$ be a quantum state and let $(\Pi, \bbI - \Pi)$ be a projective measurement on $X$ such that $\Tr(\Pi\rho) \geq 1 -\delta$. Let 
\[
    \rho' = \frac{\Pi\rho\Pi}{\Tr[\Pi\rho]}
\]
be the state after applying $(\Pi, \bbI - \Pi)$ to $\rho$ and post-selecting on obtaining the first outcome. Then $\tracedist(\rho, \rho') \leq 2\sqrt{\delta}$.
\end{lemma}

\paragraph{Quantum Fourier Transform over Finite Groups.} Let $G$ be a finite cyclic group and let $\omega_{|G|} \coloneqq \exp(\frac{2\pi i}{|G|})$ be a primitive $|G|$'th root of unity.  Let $\calX$ be a register containing a $G$-qudits. The quantum Fourier transform (QFT) over $G$ applied to $\calX$ is the operation
\[
    \ket{x}^\calX \mapsto \sum_{y\in G} \omega_{|G|}^{x y} \ket{y}^{\calX}
\]
Any Abelian group $H$ may be decomposed as a product of cyclic groups $G_1 \times \dots \times G_k$. The QFT over $H$ is the tensor of the QFTs on each of these groups. For example, if $H = G^k$, then the QFT over $H$ is given by the operation
\[
    \ket{x}^\calX \mapsto \sum_{y\in G^k} \omega_{|G|}^{x\cdot y} \ket{y}^{\calX}
\]
When we consider taking a QFT over a finite field, we technically mean taking the QFT over its additive group. For example, if $\bbF$ has order $p^k$ for some prime $p$, then its additive group is $\bbZ_{p}^k$ and the QFT is the mapping above, where $G=\bbZ_p$. 



Fourier transforms are closely related to roots of unity. An $n$'th root of unity $\omega$ is an element of $\bbC$ such that $\omega^n = 1$. $\omega$ is a \emph{primitive} $n$'th root of unity if $\omega^k \neq 1$ for every $k \in [n-1]$. The following claim about the summation of roots of unity will be useful.

\ifsubmission\begin{proposition}\else\begin{claim}\fi\label{claim:sum-roots-unity}
    Let $G$ be a cyclic group and let $\omega \neq 1$ be an $|G|$'th root of unity. Then
    \[
        \sum_{x \in G} \omega^x = 0
    \]
\ifsubmission\end{proposition}\else\end{claim}\fi
\begin{proof}
    \begin{align}
        \omega \sum_{x \in G} \omega^x &=  \sum_{x \in G} \omega^{1 + x}
        \\
        &= \sum_{z \in G} \omega^{z}
    \end{align}
    where $z = 1+x \in G$. Since $\omega \neq 1$ by definition, this can only be the case if $\sum_{x \in G} \omega^x = 0$.
\end{proof}

\subsection{Statistics}
\ifsubmission\begin{proposition}\else\begin{claim}\fi[Hoeffding's Inequality~\cite{Hoeffding1994}]\label{claim:hoeffding}
Let $X_1, \dots, X_n$ be Boolean independent random variables. Let $\mu = \bbE[\sum_{i=1}^{n} X_i]$. Then for every $\delta > 0$,
\[
    \Pr\left[\left|\sum_{i=1}^{n} X_i - \mu ]\right| \geq \delta \right] \leq 2\exp\left(\frac{-2\delta^2}{n}\right)
\]
\ifsubmission\end{proposition}\else\end{claim}\fi

\subsection{Polynomials and Reed-Solomon Codes}\label{sec:prelim-poly-RS}

We give some useful facts about polynomials over finite fields and the related Reed-Solomon error-correcting code.

\paragraph{Interpolation.} Lagrange interpolation gives a method of finding every point on a degree $d$ polynomial $f$ over any field $\bbK$, given any $d+1$ points on $f$~\cite{lagrange1901lectures}. In particular, for every degree $d$, there is a linear operation $\Interpolate_d(Y, X)$ 
that takes in a set of $d+1$ pairs $(x, f(x))$ and outputs $f(y)$ for each $y\in Y$. The operation to find a set of $s$ points is described by a matrix $R \in \bbK^{s\times (d+1)}$.

\begin{fact}\label{fact:interpolation-linear-independence}
    Let $R \in \bbK^{s \times (d+1)}$ be an interpolation matrix for a degree $d$ polynomial. Then any $s$ columns of $R$ are linearly independent.
\end{fact}
\begin{proof}
    Consider the matrix $P = [R | -I_s]$. It suffices to show that any $s$ columns of $P$ are linearly independent. Observe that $Py = 0$ if and only if $y_i$ lies on the unique degree $d$ polynomial defined by $y_1, \dots, y_{d+1}$ for every $i \in [d+2, d+1 + s]$. Assume, for the sake of contradiction, that some set of $s$ columns of $P$ were linearly dependent. Then there exists a $y$ with exactly $s$ non-zero values such that $Py = 0$. Since $y$ is in the kernel of $P$, it consists of $d+1 + s$ evaluations of a degree $d$ polynomial. However, this would imply the existence of a degree $d$ polynomial with $d +1 + s - s = d+1$ zeros, which does not exist.
\end{proof}

\paragraph{Reed-Solomon Codes.} A Reed-Solomon code is a error correcting code based on polynomials over a finite field $\bbK$~\cite{ReedSolomon60}. Given a message $m \in \bbK^{d}$, it encodes $m$ as a polynomial $f$ with degree $d+1$, then outputs $s$ evaluations of $f$ as the codeword. For finite field $\bbK$ and degree $d$, there exists an efficient correction procedure $\Correct_{\bbK,d}(X)$ which attempts to recover a polynomial $f$ using a noisy set of evaluation points $X$, e.g.~\cite{WelchBerlekamp86,Gao03}. If $X$ has size $s$ and there are $<(s-d+1)/2$ pairs $(x,y)\in X$ such that $y \neq f(x)$, then $\Correct_{\bbK,d}(X)$ outputs $f$.

\subsection{Secret Sharing}\label{subsec:prelim-ss}

\paragraph{Classical secret sharing.} We introduce the standard notion of a secret sharing scheme, which allows one party to distribute a classical secret $s$ among $n$ parties, such that only certain subsets of parties have the ability to reconstruct the secret $s$. An access structure $\bbS\subseteq \calP([n])$ for $n$ parties is a monotonic set of sets, i.e.\ if $S \in \bbS$ and $S' \supset S$, then $S' \in \bbS$. Any set of parties $S \in \bbS$ is authorized to access the secret. Secret sharing for general monotone access structures was first introduced by \cite{ItoSaiNis87}.

\begin{definition}[Secret Sharing for Monotone Access Structures]\label{def:classical-SS}
A secret sharing scheme is specified by a monotone access structure $\bbS$ over $n$ parties, and consists of two classical algorithms:
\begin{itemize}
    \item $\Split_\bbS(s)$ is a randomized algorithm that takes in a secret $s$, and outputs $n$ shares $\sh_1,\dots,\sh_n$.
    \item $\Reconstruct_\bbS(\{\sh_i\}_{i \in P})$ is a deterministic algorithm that takes in a set of shares $\{\sh_i\}_{i \in P}$ for some $P \subseteq [n]$, and outputs either $s$ or $\bot$.
\end{itemize}

The scheme should satisfy the following notions of correctness and security.

\begin{itemize}
    \item \textbf{Correctness.} For all subsets $P \subseteq [n]$ such that there exists $S \in \bbS$ such that $P \subseteq S$, 
    \[\Pr\left[\Reconstruct(\{\sh_i\}_{i \in P}) =s  : (\sh_1,\dots,\sh_n) \gets \Split_{\bbS}(s)\right] = 1.\]
    \item \textbf{Privacy.} There exists a randomized algorithm $\Sim$ such that for all subsets $P \subseteq [n]$ such that for all $S \in \bbS$, $P \not\subseteq S$, and any $s$, \[\left\{\{\sh_i\}_{i \in P} : (\sh_1,\dots,\sh_n) \gets \Split_{\bbS}(s)\right\} \equiv \left\{\{\sh_i\}_{i \in P} : \{\sh_i\}_{i \in P} \gets \Sim(P)\right\}.\]
\end{itemize}
    
\end{definition}

\paragraph{2-out-of-2 secret sharing with certified deletion} Now, we recall the definition of 2-out-of-2 secret sharing \emph{with certified deletion} as defined by \cite{C:BarKhu23}. A 2-out-of-2 secret sharing scheme is a very special case of secret sharing where the secret is split into two shares such that both shares together determine the secret, but either share individually cannot be used to recover the secret.

\begin{definition}[2-out-of-2 Secret Sharing with Certified Deletion]\label{def:2-2-SS}
We augment the standard notion of secret sharing to include a deletion algorithm $\Delete_{\Twotwo}$ and a verification algorithm $\Verify_{\Twotwo}$. We also specify that one share is quantum and the other share is classical. Finally, we introduce a security parameter $1^\secp$, since our deletion security guarantee will be statistical rather than perfect.
\begin{itemize}
    \item $\Split_{\Twotwo}(1^\secp,s)$ is a quantum algorithm that takes in the security parameter $1^\secp$, a secret $s$, and outputs a quantum share $\ket{\sh_1}$, a classical share $\sh_2$, and a classical verification key $\vk$.
    \item $\Reconstruct_{\Twotwo}(\ket{\sh_1},\sh_2)$ is a quantum algorithm that takes in two shares and outputs the secret $s$.
    \item $\Delete_{\Twotwo}(\ket{\sh_1})$ is a quantum algorithm that takes in a quantum share $\ket{\sh_1}$ and outputs a deletion certificate $\cert$.
    \item $\Verify_{\Twotwo}(\vk,\cert)$ is a classical algorithm that takes in a verification key $\vk$ and a deletion certificate $\cert$ and outputs either $\top$ or $\bot$.
\end{itemize}

Beyond satisfying the standard secret sharing notions of correctness and privacy (\cref{def:classical-SS}) for the 2-out-of-2 access structure, the scheme should satisfy the following properties.

\begin{itemize}
    \item \textbf{Deletion correctness.} For all $\secp \in \bbN$ and $s$,
    \[ \Pr\left[\Verify_{\Twotwo}(\vk,\cert) = \top : \begin{array}{r}(\ket{\sh_1},\sh_2) \gets \Split_{\Twotwo}(1^\secp,s) \\ \cert \gets \Delete_{\Twotwo}(\ket{\sh_1})\end{array}\right] = 1.\]
    \item \textbf{Certified deletion security.} Let $\Adv$ be an adversary, $s$ be a secret, and define the experiment $\TwoSSCD(1^\secp,\Adv,s)$ as follows:
    \begin{itemize}
        \item Sample $(\ket{\sh_1},\sh_2,\vk) \gets \Split_{\Twotwo}(1^\secp,s)$.
        \item Run $(\cert,\calR) \gets \Adv(1^\secp,\ket{\sh_1})$, where $\calR$ is an arbitrary output register.
        \item If $\Verify_{\Twotwo}(\vk,\cert) = \top$, output $(\calR,\sh_2)$, and otherwise output $\bot$.
    \end{itemize}
    Then, for any unbounded adversary $\Adv$ and any pair of secrets $s_0,s_1$, it holds that 
    \[\TD\left[\TwoSSCD(1^\secp,\Adv,s_0), \TwoSSCD(1^\secp,\Adv,s_1)\right] = 2^{-\Omega(\secp)}.\]
\end{itemize}

\end{definition}

\cite{C:BarKhu23} showed the existence of a 2-out-of-2 secret sharing scheme with certified deletion satisfying the above definition. Notice that our certified deletion security definition requires a trace distance of $2^{-\Omega(\secp)}$. While the theorem from \cite{C:BarKhu23} only states a bound of $\negl$, a quick inspection of their proof establishes that they in fact show a bound of $2^{-\Omega(\secp)}$.

\fi
\section{High-Rate Seedless Quantum Extractors}\label{sec:extractor}

In this section, we study seedless extraction of large amounts of entropy from a quantum source. The source of entropy comes from applying a quantum Fourier transform to a state which is ``almost'' a computational basis state. In particular, the source register $\calX$ is in superposition over vectors with low Hamming weight, and may be arbitrarily entangled with a register $\calA$ that contains side-information about the source. Previously, \cite{EC:ABKK23} showed that the XOR function perfectly extracts a single bit of entropy in this setting. However, in order to extract multiple bits of entropy, they resorted to the use of a random oracle. We also remark that the case of \emph{seeded} extraction has been well-studied by, e.g. ~\cite{TCC:RenKon05,C:DFLSS09,C:BouFeh10}.

We describe a general class of extractors that produces multiple truly random elements of any finite field $\bbF$, even conditioned on the side-information register $\calA$. In the case where the finite field has order $p^k$ for a prime $p$, we show that a large amount of entropy is generated even when the quantum Fourier transform is applied by interpreting each element $x\in \bbF_{p^k}$ as a vector $\vect{x}\in \bbF_{p}^k$ and applying the transform mod $p$ to each index (as opposed to applying the transform mod $p^k$ directly to the field element). This feature allows the source to be prepared using less entanglement in our eventual application to secret sharing.

\paragraph{Notation.} The Hamming weight $h_{\bbF}(\vect{v})$ of a vector $\vect{v} \in \bbF^M$ over a finite field $\bbF$ is its number of non-zero positions. We denote vectors $\vect{v}$ and matrices $\vect{R}$ using bold font. 
Since we will be working with elements which can be interpreted as elements of two different fields, we use $(\cdot)_{\bbF}$ to denote that the contents of the parentheses should be interpreted as elements and operations over $\bbF$.
For example, $(x + y)_{\bbF}$ denotes addition of $x$ and $y$ inside the field $\bbF$. If $\vect{x}, \vect{y} \in \bbF^k$, then $(\vect{x} + \vect{y})_\bbF$ denotes vector addition. 
For an extension field $\bbK$ of $\bbF$, a scalar $x\in \bbK$ can also be interpreted as a vector $\vect{x}\in \bbF^k$. In this case, for $x,y \in \bbK$, $(\vect{x}\cdot \vect{y})_{\bbF}$ produces a scalar in $\bbF$. If an element can be interpreted as either a vector or a scalar, we bold it depending on the context of the first operation applied; for example, $(x \vect{y})_{\bbK}$ or $(\vect{x}\cdot \vect{z})_{\bbF}$ for $x\in \bbK$, $\vect{y} \in \bbK^n$, and $\vect{z}\in \bbF^k$.



\begin{theorem}\label{lem:extractor}
    Let $\bbF$ be a finite field of order $p^k$. Let $X = \calX_1, \dots, \calX_M$ be a register containing $M$ $\bbF$-qudits, and consider any quantum state
    \[
        \ket{\gamma}^{\calA,\calX} = \sum_{\vect{u} \in \bbF^M:h_{\bbF}(\vect{u})<\frac{M-m}{2}} \ket{\psi_{\vect{u}}}^\calA \otimes \ket{\vect{u} + \vect{w}}^{\calX}
    \]
    for some integer $m \leq M$ and some fixed string $\vect{w}\in \bbF^M$.
    Let $\vect{R}\in \bbF^{m \times M}$ be a matrix such that every set of $m$ columns of $\vect{R}$ are linearly independent.
    
    Let $\rho^{\calA,\calY}$ be the mixed state that results from the following procedure:
    \begin{enumerate}
        \item Apply a quantum Fourier transform over $\bbF$'s additive group $\bbZ_{p}^k$ to each register $\calX_i$. In other words, interpret $\calX_i$ as a sequence of registers $\calX_{i,1},\dots, \calX_{i,k}$ containing $\bbZ_p$-qudits, then apply a quantum Fourier transform mod $p$ to each $\calX_{i,j}$.
        \item Initialize a fresh register $\calY$, and apply the isometry $\ket{\vect{x}}^{\calX} \mapsto \ket{\vect{x}}^{\calX}\otimes \ket{\vect{Rx}}^{\calY}$.
        \item Trace out register $\calX$.
    \end{enumerate}

    \noindent Then 
    \[
        \rho^{\calA, \calY} = \Tr^{X}[\ket{\gamma}\bra{\gamma}] \otimes \left(\frac{1}{|\bbF|^m} \sum_{\vect{y}\in \bbF^m} \ket{\vect{y}}\bra{\vect{y}} \right).
    \]
\end{theorem}


\begin{remark} As an example, consider $\bbF$ to be the field with $2^n$ elements. Note that for the source, the Hamming weight is taken over the larger field $\bbF$, but the quantum Fourier transform is done over the individual qubits, which in this case makes it just a Hadamard gate.
The extractor $R$ operates over $\bbF$ and produces an output in $\bbF^m$.
\end{remark}

\begin{proof}
    First, we apply the Fourier transform to $\ket{\gamma}$ to obtain 

    \[\sqrt{|\bbF|^{-M}} \sum_{\vect{u} \in \bbF^M:h_{\bbF}(\vect{u})<\frac{M-m}{2}} \ket{\psi_{\vect{u}}}^{\calA} \otimes \sum_{\vect{x} \in \bbF^M} \omega_{p}^{((\vect{u}+\vect{w})_{\bbF}\cdot \vect{x})_{\bbZ_p}} \ket{\vect{x}}^{\calX},\] where $\omega_{p}$ is a primitive $p$'th root of unity. Next, after applying the extractor, but before tracing out $\calX$, the state becomes
    \begin{align}
        &\sqrt{|\bbF|^{-M}} \sum_{\vect{x}\in \bbF^{M}} 
            \left(\sum_{\vect{u} \in \bbF^M:h_{\bbF}(\vect{u})<\frac{M-m}{2}} \omega_{p}^{((\vect{u}+\vect{w})_{\bbF}\cdot \vect{x})_{\bbZ_p}}\ket{\psi_{\vect{u}}}^\calA \right) 
            \otimes \ket{\vect{x}}^\calX 
            \otimes \ket{\vect{R}\vect{x}}^{\calY}
        \\
        &\coloneqq \sqrt{|\bbF|^{-M}} \sum_{\vect{x}\in \bbF^{M}} \ket{\phi_{\vect{x}}}^\calA \otimes \ket{\vect{x}}^\calX \otimes \ket{\vect{R}\vect{x}}^{\calY}.
    \end{align}
     Since the additive group of $\bbF$ is $\bbZ_p^k$, for every $\vect{u}, \vect{w}\in \bbF^M$ and $\vect{x}\in \bbZ_p^{kM}$, we have 
    \[
        ((\vect{u}+\vect{w})_{\bbF}\cdot \vect{x})_{\bbZ_p} = ((\vect{u}+\vect{w})_{\bbZ_p^k}\cdot \vect{x})_{\bbZ_p} = ((\vect{u}+\vect{w})\cdot \vect{x})_{\bbZ_p}.
    \]
    Tracing out register $\calX$ yields
    \begin{align}
        \rho^{\calA, \calY} &= |\bbF|^{-M} \sum_{\vect{x}\in \bbF^{M}} \ket{\phi_{\vect{x}}}\bra{\phi_{\vect{x}}} 
            \otimes \ket{\vect{R}\vect{x}}\bra{\vect{R}\vect{x}}
        \\
        &= |\bbF|^{-M} \sum_{\substack{\vect{y}\in \bbF^m \\ \vect{x}\in \bbF^{M}: (\vect{R}\vect{x})_{\bbF} = \vect{y}}} \ket{\phi_{\vect{x}}}\bra{\phi_{\vect{x}}} 
            \otimes \ket{\vect{y}}\bra{\vect{y}}
        \\
        &= |\bbF|^{-M} \sum_{\substack{\vect{y}\in \bbF^m \\ \vect{x}\in \bbF^{M}: (\vect{R}\vect{x})_{\bbF} = \vect{y}}} 
            \left( \sum_{\substack{\vect{u}_1, \vect{u}_2 \in \bbF^M:\\ h_{\bbF}(\vect{u}_1),h_{\bbF}(\vect{u}_2) \leq \frac{M-m}{2}}} \omega_{p}^{((\vect{u}_1+\vect{w}) \cdot \vect{x})_{\bbZ_p}} \overline{\omega_{p}^{((\vect{u}_2 + \vect{w}) \cdot \vect{x})_{\bbZ_p}}} \ket{\phi_{\vect{u}_1}}\bra{\phi_{\vect{u}_2}}\right) 
            \otimes \ket{\vect{y}}\bra{\vect{y}}
        \\
        &= \sum_{\substack{\vect{u}_1, \vect{u}_2 \in \bbF^M:\\ h_{\bbF}(\vect{u}_1),h_{\bbF}(\vect{u}_2) \leq \frac{M-m}{2}}} 
            \ket{\phi_{\vect{u}_1}}\bra{\phi_{\vect{u}_2}}
            \otimes \left(|\bbF|^{-M} \sum_{\vect{y}\in \bbF^m} \ket{\vect{y}}\bra{\vect{y}} \sum_{\vect{x}\in \bbF^{M}: (\vect{R}\vect{x})_{\bbF} = \vect{y}} \omega_{p}^{((\vect{u}_1 + \vect{w}) \cdot \vect{x} - (\vect{u}_2 + \vect{w}) \cdot \vect{x} )_{\bbZ_p}}\right)
        \\
        &= \sum_{\substack{\vect{u}_1, \vect{u}_2 \in \bbF^M:\\ h_{\bbF}(\vect{u}_1), h_{\bbF}(\vect{u}_2) \leq \frac{M-m}{2}}} 
            \ket{\phi_{\vect{u}_1}}\bra{\phi_{\vect{u}_2}}
            \otimes \left(|\bbF|^{-M} \sum_{\vect{y}\in \bbF^m} \ket{\vect{y}}\bra{\vect{y}} \sum_{\vect{x}\in \bbF^{M}: (\vect{R}\vect{x})_{\bbF} = \vect{y}} \omega_{p}^{((\vect{u}_1 - \vect{u}_2) \cdot \vect{x})_{\bbZ_p}}\right).
    \end{align}

    Next, we apply \Cref{claim:extractor-subclaim}, proven below, to show that every $\ket{\phi_{\vect{u}_1}}\bra{\phi_{\vect{u}_2}}$ term where $\vect{u}_1 \neq \vect{u}_2$ has coefficient 0. To see this, consider any such $\vect{u}_1, \vect{u}_2$ and the value $\vect{u} = (\vect{u}_1 - \vect{u}_2)_{\bbZ_p} = (\vect{u}_1 - \vect{u}_2)_{\bbF}$. Condition 1 is satisfied by $\vect{u}$ since $\vect{u}_1\neq \vect{u}_2$. Condition 2 is satisfied since $h_{\bbF}(\vect{u}) \leq h_{\bbF}(\vect{u}_1) + h_{\bbF}(\vect{u}_2) \leq M - m$, so there are at least $m$ indices of $\vect{u}$ which are zero. Finally, condition 3 is satisfied since any $m$ columns of $R$ are linearly independent. 

    Finally, noting that if $\vect{u}_1 = \vect{u}_2$, then the coefficient of $\ket{\phi_{\vect{u}_1}}\bra{\phi_{\vect{u}_1}}\otimes \ket{\vect{y}}\bra{\vect{y}}$ is the number of solutions to $(\vect{R}\vect{x})_{\bbF} = \vect{y}$, which is $|\bbF|^{M-m}$, we conclude that 
    \begin{align}
        \rho^{\calA, \calY} 
        &= \sum_{\vect{u} \in \bbF^M: h_{\bbF}(\vect{u}) \leq \frac{M-m}{2}} 
            \ket{\phi_{\vect{u}}}\bra{\phi_{\vect{u}}}
            \otimes \left(|\bbF|^{-m} \sum_{\vect{y}\in \bbF^m} \ket{\vect{y}}\bra{\vect{y}}\right)
        \\
        &= \Tr^{X}[\ket{\gamma}\bra{\gamma}]
            \otimes \left(|\bbF|^{-m} \sum_{\vect{y}\in \bbF^m} \ket{\vect{y}}\bra{\vect{y}}\right).
    \end{align}

\ifsubmission\qed\fi\end{proof}

\ifsubmission\begin{proposition}\else\begin{claim}\fi\label{claim:extractor-subclaim}
    Let $\vect{u}\in \bbF^M$ and $\vect{y}\in \bbF^m$, and suppose that 
    \begin{enumerate}
        \item $\vect{u}_i \neq 0$ for some index $i$.
        \item There exists a set $J\subseteq [0,\dots, M-1]$ of size $m$ such that for every $j\in J$, $\vect{u}_j = 0$.
        \item The submatrix $\vect{R}_J$ consisting of the columns of $\vect{R}$ corresponding to $J$ has full rank.
    \end{enumerate}
    Then
    \[
        \sum_{\vect{x}\in \bbF^M: (\vect{R}\vect{x})_\bbF = \vect{y}} \omega_{p}^{(\vect{u} \cdot \vect{x})_{\bbZ_p}} = 0.
    \]
\ifsubmission\end{proposition}\else\end{claim}\fi

\begin{remark}
We note that in the case that $\bbF = \bbF_p$, then the above expression actually holds for any $\vect{u} \notin \mathsf{rowspan}(\vect{R})$, which follows from a standard argument. The three conditions above do imply that $\vect{u} \notin \mathsf{rowspan}(\vect{R})$, but are more restrictive. We take advantage of the extra restrictions in order to prove that the expression holds even in the case where $\bbF$ is an extension field of $\bbF_p$.
\end{remark}

\begin{proof}
    Our strategy will be to partition the affine subspace $S_{\vect{y}} = \{\vect{x}\in \bbF^{M}: (\vect{R}\vect{x})_{\bbF} = \vect{y}\}$ into parallel lines, and then claim that the sum over each line is $0$. 

    To begin, define a vector $\vect{z} \in \bbF^M$ so that 
    \begin{itemize}
        \item $\vect{z}_i = 1$,
        \item $\vect{z}_j = 0$ for all $j \notin J \cup \{i\}$, and 
        \item $(\vect{R}\vect{z})_\bbF = 0^m$,
    \end{itemize}
    which is possible because the $m \times m$ submatrix $\vect{R}_J$ has full rank.
    By construction, we have that for any $c \in \bbF$, 
    \begin{equation}\label{eq:simplifyDot}
        (\vect{u} \cdot (c\vect{z})_{\bbF})_{\bbZ_p} 
        = 
        \left(\vect{u}_i \cdot (c \vect{z}_i)_{\bbF} + \sum_{j\in J} \vect{u}_{j} \cdot (c \vect{z}_j)_{\bbF} + \sum_{j\notin J \cup \{i\}} \vect{u}_{j} \cdot (0)_{\bbF}\right)_{\bbZ_p} = (\vect{u}_i \cdot \vect{c})_{\bbZ_p},
    \end{equation}
    where note that in the final expression, $\vect{u}_i$ and $\vect{c}$ are interpreted as vectors in $\bbZ_p^k$.
    
    
    
    
    
    Now, fix any $\vect{x} \in S_{\vect{y}}$ and $c \in \bbF$. Then we have that $(\vect{x} + c\vect{z})_{\bbF} \in S_{\vect{y}}$, since $(\vect{R}(\vect{x} + c\vect{z}))_{\bbF} = \vect{y} + 0$. Therefore, we can partition $S_{\vect{y}}$ into one-dimensional cosets (lines) of the form $\{\vect{x} + c\vect{z}\}_{c \in \bbF}$.

    We now show that the sum over any particular coset is $0$, i.e. that for any $\vect{x} \in S_{\vect{y}}$,
    \[
        \sum_{c \in \bbF} \omega_{p}^{(\vect{u} \cdot (\vect{x} + c \vect{z})_{\bbF})_{\bbZ_p}} = 0.
    \]

    Since the additive group of $\bbF$ is $\bbZ_p^k$, by \cref{eq:simplifyDot} we have that
    \begin{align*}
        (\vect{u} \cdot (\vect{x} + c \vect{z})_{\bbF})_{\bbZ_p} 
            &= (\vect{u} \cdot \vect{x} + \vect{u} \cdot (c \vect{z})_{\bbF})_{\bbZ_p}
            \\
            &= (\vect{u} \cdot \vect{x} + \vect{u}_{i} \cdot \vect{c})_{\bbZ_p}
    \end{align*}
    We now view $\vect{u}_i\in \bbF$ as a vector $\vect{u}_i = u_{i,0},\  \dots,\  u_{i,k-1} \in \bbZ_{p}^k.$
    In particular, since $\vect{u}_i\neq 0$, there exists an index $t$ such that $u_{k,t} \neq 0 \in \bbZ_p$. 
    By also interpreting $\vect{c}\in \bbF$ as an element of $\bbZ_{p}^{k}$, we can decompose it as $\vect{c} = (\vect{c'} + c_t \vect{e}_t)_{\bbZ_p}$, where $\vect{c'} \in \bbZ_{p}^{k}$ is such that $\vect{c}'_t = 0$, where $c_t \in \bbZ_p$, and where $\vect{e_t}\in \bbZ_{p}^{k}$ is the $t$'th standard basis vector. Therefore 
    \[
        (\vect{u} \cdot (\vect{x} + c \vect{z})_{\bbF})_{\bbZ_p} = (\vect{u} \cdot \vect{x} + \vect{u}_i \cdot \vect{c'} + u_{i,t} c_t)_{\bbZ_p}.
    \]

    Since $u_{i,t} \neq 0$ and $\omega_{p}$ is a primitive $p$'th root of unity, we know that $\omega_{p}^{u_{i,t}} \neq 1$ is a $p$'th root of unity. Therefore by \Cref{claim:sum-roots-unity},
    \begin{align}
        \sum_{c \in \bbF} 
            \omega_{p}^{(\vect{u} \cdot (\vect{x} + \vect{c} \vect{z})_{\bbF})_{\bbZ_p}}
        &= \sum_{\vect{c'} \in \bbZ_{p}^{k}:\vect{c}'_t = 0}  \omega_{p}^{(\vect{u} \cdot \vect{x} + \vect{u}_i \cdot \vect{\vect{c}'})_{\bbZ_{p}}} \cdot 
            \sum_{c_t \in \bbZ_p} \left(\omega_{p}^{u_{i,t}}\right)^{c_t}
        \\
        &= \sum_{\vect{c}' \in \bbZ_{p}^{k}:\vect{c}'_t = 0} \omega_{p}^{(\vect{u} \cdot \vect{x} + \vect{u}_i \cdot \vect{c}')_{\bbZ_{p}}}  \cdot 0\\
        &= 0.
    \end{align}
\ifsubmission\qed\fi\end{proof}
\section{Definitions of Secret Sharing with Certified Deletion}\label{sec:definitions}

A secret sharing scheme with certified deletion augments the syntax of a secret sharing scheme with additional algorithms to delete shares and verify deletion certificates. We define it for general access structures. As described in \cref{subsec:prelim-ss}, an access structure $\bbS\subseteq \calP([n])$ for $n$ parties is a monotonic set of sets, i.e.\ if $S \in \bbS$ and $S' \supset S$, then $S' \in \bbS$. Any set of parties $S \in \bbS$ is authorized to access the secret. A simple example of an access structure is the threshold structure, where any set of at least $k$ parties is authorized to access the secret. We denote this access structure as $(k, n)$.

\begin{definition}[Secret Sharing with Certified Deletion]\label{def:sscd-syntax}
    A secret sharing scheme with certified deletion is specified by a monotone access structure $\bbS$ over $n$ parties, and consists of four algorithms:
    \begin{itemize}
        \item $\Split_\bbS(1^\secpar, s)$ takes in a secret $s$, and outputs $n$ share registers $\calS_1, \dots, \calS_n$ and a verification key $\verkey$.
        \item $\Reconstruct_\bbS(\{S_i\}_{i\in P})$ takes in set of share registers for some $P \subseteq [n]$, and outputs either $s$ or $\bot$.
        \item $\Delete_\bbS(\calS_i)$ takes in a share register and outputs a certificate of deletion $\cert$.
        \item $\Verify_\bbS(\verkey, i, \cert)$ takes in the verification key $\verkey$, an index $i$, and a certificate of deletion $\cert$, and outputs either $\top$ (indicating accept) or $\bot$ (indicating reject).
    \end{itemize}
\end{definition}

\begin{definition}[Correctness of Secret Sharing with Certified Deletion]\label{def:ss-recon-correct}
    A secret sharing scheme with certified deletion must satisfy two correctness properties:
    \begin{itemize}
        \item \textbf{Reconstruction Correctness.} For all $\secpar \in \bbN$ and all sets $S\in \bbS$,
        \[
            \Pr\left[\Reconstruct_\bbS(\{\calS_i\}_{i\in S}) : (\calS_1, \dots, \calS_n, \verkey) \gets \Split_\bbS(1^\secpar,s)\right] = 1.
        \]
        
        \item \textbf{Deletion Correctness.} For all $\secpar \in \bbN$ and all $i \in [n]$,
        \[
            \Pr\left[\Verify_\bbS(\vk, i, \cert) = \top : 
                \begin{array}{r}
                    (\calS_1, \dots, \calS_n, \verkey) \gets \Split_\bbS(1^\secp, s) \\
                    \cert \gets \Delete_\bbS(\calS_i)
                \end{array}
            \right] = 1.
        \]
        
    \end{itemize}
\end{definition}

The standard notion of security for secret sharing requires that no set of unauthorized shares $S\notin \bbS$ reveals any information about the secret (see \cref{subsec:prelim-ss}). We next present our notion of \emph{no-signaling certified deletion security}. Here, the shares are partitioned into unauthorized sets, and different parts of the adversary operate on each partition, potentially deleting some number of shares from each. The different parts of the adversary are allowed to share entanglement, but are not allowed to signal. If the adversaries jointly produce a valid certificate for at least one share from every authorized set, then we require that the joint residual state of \emph{all} of the adversaries contains no (or negligible) information about the secret. Observe that this notion of security is at least as strong as the standard notion of security for secret sharing (if we relax to statistical rather than perfect security). Indeed, if the standard notion does not hold, and thus there is some unauthorized set $S$ that leaks information about the secret, then the adversary would be able to win the certified deletion game by honestly deleting every share except for those in $S$.




\begin{definition}[No-Signaling Certified Deletion Security for Secret Sharing]\label{def:ns-cd}
    Let $P = (P_1,\dots,P_\ell)$ be a partition of $[n]$, let $\ket{\psi}$ be an $\ell$-part state on registers $\calR_1,\dots,\calR_\ell$, and let $\Adv = (\Adv_1,\dots,\Adv_\ell)$ be an $\ell$-part adversary. Define the experiment $\NSCD_\bbS(1^\secp,P,\ket{\psi},\Adv,s)$ as follows:
    \begin{enumerate}
        \item Sample $(\calS_1,\dots,\calS_n,\vk) \gets \Split_\bbS(1^\secp,s)$.
        \item For each $t \in [\ell]$, run $(\{\cert_i\}_{i \in P_t},\calR_t') \gets \Adv_t(\{\calS_i\}_{i \in P_t},\calR_t)$, where $\calR_t'$ is an arbitrary output register.
        \item If for all $S \in \mathbb{S}$, there exists $i \in S$ such that $\Verify_\bbS(\vk,i,\cert_i) = \top$, then output $(\calR'_1,\dots,\calR'_\ell)$, and otherwise output $\bot$.
    \end{enumerate}

    A secret sharing scheme for access structure $\bbS$ has no-signaling certified deletion security if for any ``admissible'' partition $P = (P_1,\dots,P_\ell)$ (i.e.\ for all $P_t \in P$ and $S \in \mathbb{S}$, $P_t \not\subseteq S$), any $\ell$-part state $\ket{\psi}$, any (unbounded) $\ell$-part adversary $\Adv$, and any pair of secrets $s_0,s_1$,
    \[
        \TD[\NSCD_\bbS(1^\secp,P,\ket{\psi},\Adv,s_0),\  \NSCD_\bbS(1^\secp,P,\ket{\psi},\Adv,s_1)] = \negl.
    \]
\end{definition}

Next, we present an alternative definition which allows the adversary to start by corrupting some unauthorized set, and then continue to adaptively delete some shares and corrupt new parties, as long as the total set of parties corrupted minus the set of shares deleted is unauthorized. 
Similarly to the previous definition, adaptive certified deletion for secret sharing subsumes the standard notion of security for secret sharing.

\begin{definition}[Adaptive Certified Deletion for Secret Sharing]\label{def:adapt-CD}
    Let $\Adv$ be an adversary with internal register $\calR$ which is initialized to a state $\ket{\psi}$, let $\bbS$ be an access structure, and let $s$ be a secret. Define the experiment $\ACD_{\bbS}(1^\secpar, \ket{\psi}, \Adv, s)$ as follows:
    \begin{enumerate}
        \item Sample $(\calS_1, \dots, \calS_n, \verkey) \gets \Split_{\bbS}(1^\secpar, s)$. Initialize the corruption set $C = \emptyset$ and the deleted set $D = \emptyset$.
        
        \item In each round $i$, the adversary may do one of three things:
        \begin{enumerate}
            \item End the experiment by outputting a register $\calR \gets \Adv(\{\calS_j\}_{j\in C}, \calR)$.
            
            \item Delete a share by outputting a certificate $\cert_i$, an index $j_i \in [n]$, and register $(\cert_{i}, j_i,\calR) \gets \Adv(\{\calS_j\}_{j\in C}, \calR)$. When the adversary chooses this option, if $\Verify_\bbS(\verkey, j_i, \cert_i)$ outputs $\top$, then add $j_i$ to $D$. Otherwise, immediately abort the experiment and output $\bot$.
            
            \item Corrupt a new share by outputting an index $j_i \in [n]$ and register $(j_i, \calR) \gets \Adv(\{\calS_j\}_{j\in C}, \calR)$. When the adversary chooses this option, add $j_i$ to $C$. If $C\backslash D \in \bbS$, immediately abort the experiment and output $\bot$.
        \end{enumerate}
        \item Output $\calR$, unless the experiment has already aborted.
    \end{enumerate}

    A secret sharing scheme for access structure $\bbS$ has adaptive certified deletion security if for any (unbounded) adversary $\Adv$, any state $\ket{\psi}$, and any pair of secrets $(s_0, s_1)$,
    \[
        \TD[\ACD_{\bbS}(1^\secpar, \ket{\psi}, \Adv, s_0),\ \ACD_{\bbS}(1^\secpar,\ket{\psi}, \Adv, s_1)] = \negl
    \]
\end{definition}

It will also be convenient to establish some notation for the order of the corrupted and deleted shares.  Let $c_{\corctr}$ be the $\corctr$'th share to be corrupted (i.e. added to $C$) and let $d_{\delctr}$ be the $\delctr$'th share to be deleted (i.e. added to $D$).






\section{Secret Sharing with No-Signaling Certified Deletion}\label{sec:no-signaling-construction}

In this section, we'll show how to combine any classical secret sharing scheme $(\CSplit_\bbS, \CReconstruct_\bbS)$ (\cref{def:classical-SS}) for access structure $\bbS \in \calP([n])$ with a 2-out-of-2 secret sharing scheme with certified deletion $(\Split_{\Twotwo},\allowbreak\Reconstruct_{\Twotwo},\allowbreak\Delete_{\Twotwo},\allowbreak\Verify_{\Twotwo})$ (\cref{def:2-2-SS}) in order to obtain a secret sharing scheme for $\bbS$ that satisfies no-signaling certified deletion security (\cref{def:ns-cd}).

\begin{figure}[h]
\begin{framed}
\begin{minipage}{\textwidth}

\paragraph{$\underline{\Split_\bbS(1^\secpar, s)}$} 
\begin{itemize}
    \item Sample $(\sh_1,\dots,\sh_n) \gets \CSplit_\bbS(s)$.
    \item Set $\kappa = \max\{\secp,n\}^2$, and for each $i \in [n]$, sample $(\ket{\qsh_i},\csh_i,\vk_i) \gets \Split_{\Twotwo}(1^\kappa,\sh_i)$.
    \item For each $i \in [n]$, sample $(\csh_{i,1},\dots,\csh_{i,n}) \gets \CSplit_{\bbS}(\csh_i)$.
    \item Set $\vk = (\vk_1,\dots,\vk_n)$, and initialize register $\calS_i$ to $\ket{\qsh_i},\{\csh_{j,i}\}_{j \in [n]}.$
\end{itemize}

\paragraph{$\underline{\Reconstruct_\bbS(\{\calS_i\}_{i \in P})}$}
\begin{itemize}
    \item Parse each register $\calS_i$ as $\ket{\qsh_i},\{\csh_{j,i}\}_{j \in [n]}.$
    \item For each $i \in P$, compute $\csh_{i} \gets \CReconstruct_\bbS(\{\csh_{i,j}\}_{j \in P})$, and output $\bot$ if the result is $\bot$.
    \item For each $i \in P$, compute $\sh_i \gets \Reconstruct_{\Twotwo}(\ket{\qsh_i},\csh_i)$.
    \item Output $s \gets \CReconstruct_\bbS(\{\sh_i\}_{i \in P})$.
\end{itemize}

\paragraph{$\underline{\Delete_\bbS(\calS_i)}$}
\begin{itemize}
    \item Parse $\calS_i$ as $\ket{\qsh_i},\{\csh_{j,i}\}_{j \in [n]}$ and output $\cert \gets \Delete_{\Twotwo}(\ket{\qsh_i})$.
\end{itemize}

\paragraph{$\underline{\Verify_\bbS(\vk,i,\cert)}$}
\begin{itemize}
    \item Parse $\vk = (\vk_1,\dots,\vk_n)$ and output $\Verify_{\Twotwo}(\vk_i,\cert)$.
\end{itemize}

\end{minipage}
\end{framed}
\caption{Secret sharing with no-signaling certified deletion security for any access structure $\bbS$.}\label{fig:SS-CD-general}
\end{figure}

\begin{theorem}
    The construction given in \cref{fig:SS-CD-general} satisfies no-signaling certified deletion security (\cref{def:ns-cd}).
\end{theorem}

\begin{proof}

Let $\Adv = (\Adv_1,\dots,\Adv_\ell)$ be any $\ell$-part adversary that partitions the shares using an admissible partition $P = (P_1,\dots,P_\ell)$ and is initialized with the $\ell$-part state $\ket{\psi}$ on registers $\calR_1,\dots,\calR_\ell$. Let $s_0,s_1$ be any two secrets, and assume for contradiction that 

\[\TD[\NSCD_\bbS(1^\secp,P,\ket{\psi},\Adv,s_0),\  \NSCD_\bbS(1^\secp,P,\ket{\psi},\Adv,s_1)] = \mathsf{nonnegl}(\secp).\]

Now, for $s \in \{s_0,s_1\}$, define a hybrid $\calH_1(s)$ as follows.\\

\noindent\underline{$\calH_1(s)$}
\begin{enumerate}
    \item Sample $C \gets \calP([n])$.
    \item Sample $(\sh_1,\dots,\sh_n) \gets \CSplit_\bbS(s)$.
    \item Set $\kappa = \max\{\secp,n\}^2$, and for each $i \in [n]$, sample $(\ket{\qsh_i},\csh_i,\vk_i) \gets \Split_{\Twotwo}(1^\kappa,\sh_i)$.
    \item For each $i \in [n]$, sample $(\csh_{i,1},\dots,\csh_{i,n}) \gets \CSplit_{\bbS}(\csh_i)$.
    \item Set $\vk = (\vk_1,\dots,\vk_n)$, and initialize register $\calS_i$ to $\ket{\qsh_i},\{\csh_{j,i}\}_{j \in [n]}.$
    \item For each $t \in [\ell]$, run $(\{\cert_i\}_{i \in P_t},\calR'_t) \gets \Adv_t(\{\calS_i\}_{i \in P_t},\calR_t)$.
    \item Let $C^* \coloneqq \{i : \Verify_{\Twotwo}(\vk_i,i,\cert_i) = \top\}$. Output $(\calR'_1,\dots,\calR'_\ell)$ if $C = C^*$ and $[n] \setminus C^* \notin \bbS$, and otherwise output $\bot$.
\end{enumerate}

That is, $\calH_1(s)$ is the same as $\NSCD_\bbS(1^\secp,P,\ket{\psi},\Adv,s)$ except that $\calH_1(s)$ makes a uniformly random guess $C$ for the subset of shares for which the adversary produces a valid deletion certificate, and aborts if this guess is incorrect.

\ifsubmission\begin{proposition}\else\begin{claim}\fi\label{claim:H1}
$\TD\left[\calH_1(s_0),\calH_1(s_1)\right] = \mathsf{nonnegl}(\secp) \cdot 2^{-n}.$ 
\ifsubmission\end{proposition}\else\end{claim}\fi

\begin{proof}
This follows directly from the fact that $\calH_1(s)$'s guess for $C$ is correct with probability $1/2^n$, and, conditioned on the guess being correct,  $\calH_1(s)$ is identical to $\NSCD_\bbS(1^\secp,P,\ket{\psi},\Adv,s)$.
\ifsubmission\qed\fi\end{proof}

Now, for $s \in \{s_0,s_1\}$ and $k \in [0,\dots,n]$, define a sequence of hybrids $\calH_{2,k}(s)$ as follows.\\

\noindent\underline{$\calH_{2,k}(s)$}
\begin{enumerate}
    \item Sample $C \gets \calP([n])$.
    \item Sample $(\sh_1,\dots,\sh_n) \gets \CSplit_\bbS(s)$.
    \item Set $\kappa = \max\{\secp,n\}^2$ and for each $i \in [n]$, if $i \leq k$ and $i \in C$, sample $(\ket{\qsh_i},\csh_i,\vk_i) \gets \Split_{\Twotwo}(1^\kappa,\bot)$, and otherwise sample $(\ket{\qsh_i},\csh_i,\vk_i) \gets \Split_{\Twotwo}(1^\kappa,\sh_i)$.
    \item For each $i \in [n]$, sample $(\csh_{i,1},\dots,\csh_{i,n}) \gets \CSplit_\bbS(\csh_i)$.
    \item Set $\vk = (\vk_1,\dots,\vk_n)$, and initialize register $\calS_i$ to $\ket{\qsh_i},\{\csh_{j,i}\}_{j \in [n]}$.
    \item For each $t \in [\ell]$, run $(\{\cert_i\}_{i \in P_t},\calR'_t) \gets \Adv_t(\{\calS_i\}_{i \in P_t},\calR_t)$.
    \item Let $C^* \coloneqq \{i : \Verify_{\Twotwo}(\vk_i,i,\cert_i) = \top\}$. Output $(\calR'_1,\dots,\calR'_\ell)$ if $C = C^*$ and $[n] \setminus C^* \notin \bbS$, and otherwise output $\bot$.
\end{enumerate}

First, note that $\calH_1(s_0) = \calH_{2,0}(s_0)$ and $\calH_1(s_1) = \calH_{2,0}(s_1)$. Next, we show the following claim.

\ifsubmission\begin{proposition}\else\begin{claim}\fi\label{claim:H2n}
$\calH_{2,n}(s_0) \equiv \calH_{2,n}(s_1)$.
\ifsubmission\end{proposition}\else\end{claim}\fi

\begin{proof}

In each experiment, if the output is not $\bot$, then we know that $[n] \setminus C$ is an unauthorized set. Moreover, the experiments do not depend on the information $\{\sh_i\}_{i \in C}$. Thus, the claim follows by the perfect privacy of $(\CSplit_\bbS,\CReconstruct_\bbS)$, which implies that 

\begin{align*}
&\{\{\sh_i\}_{i \in C \setminus [n]}: (\sh_1,\dots,\sh_n) \gets \CSplit_\bbS(s_0)\}\\ &~~~~\equiv \{\{\sh_i\}_{i \in C \setminus [n]}: (\sh_1,\dots,\sh_n) \gets \CSplit_\bbS(s_1)\}.
\end{align*}

\ifsubmission\smallskip\qed\fi\end{proof}

Finally, we show the following claim.

\ifsubmission\begin{proposition}\else\begin{claim}\fi\label{claim:H2k}
For $s \in \{s_0,s_1\}$ and $k \in [n]$, it holds that $\TD\left[\calH_{2,k-1}(s),\calH_{2,k}(s)\right] = 2^{-\Omega(\kappa)}.$
\ifsubmission\end{proposition}\else\end{claim}\fi

\begin{proof}
The only difference between these hybrids is that if $k \in C$, we switch $\sh_k$ to $\bot$ in the third step. So, suppose that $k \in C$, and consider the following reduction to the certified deletion security (\cref{def:2-2-SS}) of $(\Split_{\Twotwo},\allowbreak\Reconstruct_{\Twotwo},\allowbreak\Delete_{\Twotwo},\allowbreak\Verify_{\Twotwo})$. This experiment is parameterized by a bit $b$ which determines which one of two secrets the certified deletion challenger will share. 

\begin{itemize}
    \item The reduction samples $C \gets \calP([n])$ and $(\sh_1,\dots,\sh_n) \gets \CSplit_\bbS(s)$, and sends $\{\sh_k,\bot\}$ to the challenger.
    \item The challenger samples $(\ket{\qsh_k},\csh_k,\vk_k) \gets \Split_{\Twotwo}(1^\kappa,\sh_k)$ if $b=0$ or $(\ket{\qsh_k},\csh_k,\vk_k) \gets \Split_{\Twotwo}(1^\kappa,\bot)$ if $b=1$, and sends $\ket{\qsh_k}$ to the reduction.
    \item For each $i \in [n] \setminus \{k\}$, if $i < k$ and $i \in C$, the reduction samples $(\ket{\qsh_i},\csh_i,\vk_i) \gets \Split_{\Twotwo}(1^\kappa,\bot)$, and otherwise samples $(\ket{\qsh_i},\csh_i,\vk_i) \gets \Split_{\Twotwo}(1^\kappa,\sh_i)$.
    \item Let $t^* \in [\ell]$ be such that $k \in P_{t^*}$. For each $i \in [n] \setminus \{k\}$, the reduction samples $(\csh_{i,1},\dots,\csh_{i,n}) \gets \CSplit_\bbS(\csh_i)$. Next, the reduction samples $\{\csh_{k,i}\}_{i \in P_{t^*}} \gets \Sim(P_{t^*})$, where $\Sim$ is the simulator guaranteed by the privacy of $(\CSplit_\bbS,\CReconstruct_\bbS)$.
    \item For each $i \in P_{t^*}$, initialize register $\calS_i$ to $\ket{\qsh_i},\{\csh_{j,i}\}_{j \in [n]}$.
    \item The reduction runs $(\{\cert_i\}_{i \in P_{t^*}},\calR'_{t^*}) \gets \Adv_{t^*}(\{\calS_i\}_{i \in P_{t^*}},\calR_{t^*})$, and sends $\cert_k$ to the challenger.
    \item The challenger checks whether $\Verify_{\Twotwo}(\vk_k,\cert_k) = \top$. If so, the challenger returns $\csh_k$, and otherwise the experiment aborts and outputs $\bot$.
    \item The reduction samples $\{\csh_{k,i}\}_{i \in [n] \setminus P_{t^*}}$ conditioned on the joint distribution of $(\csh_{k,1},\dots,\csh_{k,n})$ being identical to $\CSplit_\bbS(\csh_k)$. This is possible due to the guarantee of $\Sim(P_{t^*})$, that is, \begin{align*}
    &\left\{\{\csh_{k,i}\}_{i \in P_{t^*}} : (\csh_{k,1},\dots,\csh_{k,n}) \gets \CSplit_{\bbS}(\csh_k)\right\}\\ &~~~~\equiv \left\{\{\csh_{k,i}\}_{i \in P_{t^*}} : \{\csh_{k,i}\}_{i \in P_{t^*}} \gets \Sim(P_{t^*})\right\}.
    \end{align*}
    \item For each $i \in [n] \setminus P_{t^*}$, the reduction initializes register $\calS_i$ to $\ket{\qsh_i}, \{\csh_{j,i}\}_{j \in [n]}$.
    \item For each $t \in [\ell] \setminus \{t^*\}$, run $(\{\cert_i\}_{i \in P_t},\calR'_t) \gets \Adv_t(\{\calS_i\}_{i \in P_t},\calR_t)$.
    \item Let $C^* \coloneqq \{i : \Verify_{\Twotwo}(\vk_i,i,\cert_i) = \top\}$. The reduction outputs $(\calR'_1,\dots,\calR'_\ell)$ if $C = C^*$ and $[n] \setminus C^* \notin \bbS$, and otherwise outputs $\bot$.
\end{itemize}

Observe that in the case $b=0$, the output of this experiment is identical to $\calH_{2,k-1}(s)$ while if $b=1$, the output of this experiment is identical to $\calH_{2,k}(s)$. Thus, the claim follows from the certified deletion security of $(\Split_{\Twotwo},\Reconstruct_{\Twotwo},\Delete_{\Twotwo},\Verify_{\Twotwo})$.

\ifsubmission\qed\fi\end{proof}

Thus, by combining \cref{claim:H2n} and \cref{claim:H2k}, we have that \[\TD\left[\calH_1(s_0),\calH_1(s_1)\right] = 2n \cdot 2^{-\Omega(\kappa)} \leq 2^{-\Omega(\{\max\{\secp,n\}^2\})}.\] However, this violates \cref{claim:H1}, since \[2^{-\Omega(\{\max\{\secp,n\}^2\})} < \mathsf{nonnegl}(\secp) \cdot 2^{-n},\] which completes the proof.

\ifsubmission\qed\fi\end{proof}

\section{Threshold Secret Sharing with Adaptive Certified Deletion}\label{sec:adaptive-construction}

In this section, we show how to construct a secret sharing scheme for threshold access structures that satisfies adaptive certified deletion (\Cref{def:adapt-CD}).

\subsection{Construction}\label{subsec:adaptive_construction}

Our construction is given in \Cref{fig:SS-CD-Construction}, which uses a set of parameters described in \Cref{fig:acd-parameters-loose}. We provide some intuition about the parameter settings here.

The secret is encoded in a polynomial $f$ of degree $p$. For security, we need $p$ to be at least as large as the number of points of $f$ that the adversary can learn. At most, the adversary can hold up to $k-1$ intact shares and the residual states of $n-k+1$ deleted shares. Each of the $k-1$ intact shares contains $t'$ evaluations of $f$. Additionally, the adversary may retain some small amount of information about each of the deleted shares. We upper bound the retained information by a parameter $\ell$ for each share. This gives the adversary a maximum of
\[
    (k-1)t' + (n-k+1)\ell
\]
evaluations of $f$, which becomes the minimum safe setting for $p$.

Each share will also include a number of ``check positions'', which contain Fourier basis states that are used for verification of deletion. The number of check positions $r$ and upper bound $\ell$ are set roughly so that with overwhelming probability, the adversary can retain no more than $\ell$ evaluations of $f$ when it deletes a share (more precisely, the adversary may retain a \emph{superposition} over potentially different sets of $\ell$ evaluations). 
The reader may find it useful to think of $\ell$ as being the maximum number of unexamined positions in a classical string $x$ when an adversary successfully creates a string $y$ that matches $x$ on $r$ random verification indices. 
Finally, the total size $t$ of each share is set so that $k$ shares contain less than $(kt-p)/2$ check positions, which is the maximum number of errors that can be corrected in $kt$ evaluations of a polynomial of degree $p$ (see \Cref{sec:prelim-poly-RS}).

\begin{figure}[t]
    \begin{framed}
    \begin{minipage}{\textwidth}
         The construction in \Cref{fig:SS-CD-Construction} uses the following parameters.\footnote{The parameters provided here are slightly looser than necessary, to facilitate easier inspection. We present a tighter set of parameters in \Cref{fig:acd-parameters-tight}.}
         \begin{itemize}
            \item Each share consists of $t$ total $\bbK$-registers, where 
            \[
            t = (k+1)\numChecks \left(1 + \frac{(n-k+1)\log(\secpar)}{\secpar}\right) + 1
            \]
            
            \item A share is divided into $\numChecks$ check indices and $t'=t-r$ data indices, where
            \[
            \numChecks = (\secpar + (n-k+1)\log(\secpar))^2
            \]

            \item $\ell$ intuitively represents an upper bound on the amount of information which is not destroyed when an adversary generates a valid deletion certificate for a share.
            \[
                \ell = \delLossValueShort
            \]

            See the proof of \Cref{claim:strongcd} for a more precise usage of $\ell$.
            
            \item The secret will be encoded in a polynomial of degree
            \[
            p = (k-1)t' + (n-k+1)\delLoss
            \]
         \end{itemize}
    \end{minipage}
    \end{framed}
    \caption{Parameters for Secret Sharing with Adaptive Certified Deletion}
    \label{fig:acd-parameters-loose}
\end{figure}

\begin{figure}[ht!]
\begin{framed}
\begin{minipage}{\textwidth}
\textbf{\underline{Parameters:}} Let $\bbF_2$ be the binary field and let $\bbK$ be the field with $2^{\lceil \log_2(nt+1) \rceil}$ elements. See \Cref{fig:acd-parameters-loose} for descriptions and settings of the parameters $t$, $t'$, $r$, $\ell$, and $p$.

\paragraph{$\underline{\Split_{(k, n)}(1^\secpar, s)}$} 
\begin{itemize}
    \item Sample a random polynomial $f$ with coefficients in $\bbK$ and degree $p$ such that $f(0) = s$.
    \item For each $i\in [n]$:
    \begin{enumerate}
        \item Sample a random set of indices $J_i \subset [t]$ of size $t'= t-r$.
        \item For each $j\in J_i$, set $\ket{\psi_{i, j}} = \ket{f(it + j)}$. These are the $t'$ data positions.
        \item For each $j\in [t]\backslash J_i$, sample a uniform element $y_{i,j} \gets \bbK$ and set $\ket{\psi_{i, j}} = H^{\otimes \lceil \log_2(n + 1) \rceil} \ket{y_{i,j}}$.
        These are the $r$ check positions.
        \item Initialize register $\calS_i$ to $\bigotimes_{j=1}^{t} \ket{\psi_{i, j}}$.

    \end{enumerate}
    \item Set $\vk = \{J_i, \{y_{i,j}\}_{j\in [t]\backslash J_i}\}_{i\in [n]}$.
\end{itemize}

\paragraph{$\underline{\Reconstruct_{(k, n)}(1^\secpar, \{\calS_{i}\}_{i\in P})}$} 
\begin{itemize}
    \item If $|P| < k$, output $\bot$. Otherwise, set $P'$ to be any $k$ shares in $P$.
    \item For each $i \in P'$, measure $\calS_i$ in the computational basis to obtain $y_i = (y_{i,1},\dots, y_{i,t}) \in \bbK^t$.
    \item Compute $f \gets \Correct_{\bbK, p}(\{(it + j, y_{i,j})\}_{i\in P', j\in [t]})$, as defined in \Cref{sec:prelim-poly-RS}.
    \item Output $f(0)$.
\end{itemize}

\paragraph{$\underline{\Delete_{(k,n)}(1^\secpar, \calS_i)}$} 
\begin{itemize}
    \item Parse $\calS_i$ as a sequence of $t\lceil \log_2(n+1) \rceil$ single qubit registers, measure each qubit register in Hadamard basis and output the result $\cert$.
\end{itemize}

\paragraph{$\underline{\Verify_{(k,n)}(1^\secpar, \verkey, i, \cert)}$}
\begin{itemize}
    \item Parse $\verkey = \{J_i, \{y_{i,j}\}_{j\in [t]\backslash J_i}\}_{i\in [n]}$, and parse $\cert \in \bbK^t$ as a sequence of $t$ elements of $\bbK$. Output $\top$ if $\cert_j = y_{i,j}$ for every $j \in J_i$, and $\bot$ otherwise.
\end{itemize}

\end{minipage}
\end{framed}
\caption{Construction for Secret Sharing with Adaptive Certified Deletion}\label{fig:SS-CD-Construction}
\end{figure}

\begin{theorem}\label{thm:ss-strongcd}
    There exists secret sharing for threshold access structures which satisfies adaptive certified deletion.
\end{theorem}
\begin{proof}
    \ifsubmission
    The construction is given in \Cref{fig:SS-CD-Construction}. Deletion correctness is apparent from inspection of the construction. We prove reconstruction correctness in \Cref{claim:recon-correctness}. See \Cref{sec:overview-proving-adaptive} for an overview of the security proof and \Cref{app:acd-proof} for more details.
    \else
    The construction is given in \Cref{fig:SS-CD-Construction}. Deletion correctness is apparent from inspection of the construction. We prove reconstruction correctness in \Cref{claim:recon-correctness} and adaptive certified deletion security in \Cref{claim:strongcd}.
    \fi
\ifsubmission\qed\fi\end{proof}

\begin{lemma}\label{claim:recon-correctness}
    The construction in \Cref{fig:SS-CD-Construction} using parameters from \Cref{fig:acd-parameters-loose} has reconstruction correctness.
\end{lemma}
\begin{proof}
    The set $\{(it + j, y_{i,j})\}_{i\in P', j\in [t]}$ contains $kt$ pairs which were obtained by measuring $k$ shares. As mentioned in \Cref{sec:prelim-poly-RS}, if all but $e<(kt - p)/2$ of these pairs $(it+j, y_{i,j})$ satisfy $y_{i,j} = f(it + j)$, then $\mathsf{Correct}_{\bbK,p}$ recovers the original polynomial $f$, where $f(0) = s$. The only points which do not satisfy this are the check positions, of which there are $\numChecks$ per share, for a total of $k\numChecks$. Therefore for correctness, we require that
    
    \begin{align}
    2k\numChecks &< kt - p
        \\
        &= kt - (k-1)(t-\numChecks) - (n-k+1)\delLoss
        \\
        &= t + (k-1)\numChecks- (n-k+1)\delLoss
    \end{align}
    Therefore $t - (n-k+1)\ell > (k+1)\numChecks$. Substituting $\delLoss = \delLossValueShort$ yields
    \begin{equation}
         t\left(1 - (n-k+1)\frac{\log(\secpar)}{\sqrt{r}}\right) > (k+1)\numChecks
     \end{equation}
    \begin{align}
        t &> (k+1)\numChecks\frac{1}{1 - (n-k+1)\frac{\log(\secpar)}{\sqrt{r}}}
        \label{eq:correctness-nonnegative}
        \\
        &= (k+1)\numChecks\frac{\sqrt{\numChecks}}{\sqrt{\numChecks} - (n-k+1)\log(\secpar)}
        \\
        &= (k+1)\numChecks \left(1 + \frac{(n-k+1)\log(\secpar)}{\sqrt{\numChecks} - (n-k+1)\log(\secpar)}\right)
        \\
        &= (k+1)\numChecks \left(1 + \frac{(n-k+1)\log(\secpar)}{\secpar + (n-k+1)\log(\secpar) - (n-k+1)\log(\secpar)}\right)
        \\
        &= (k+1)\numChecks \left(1 + \frac{(n-k+1)\log(\secpar)}{\secpar}\right)
    \end{align}
    Note that \Cref{eq:correctness-nonnegative} requires that $\left(1 - (n-k+1)\delLossValueShort\right) > 0$. Since the number of check positions is $\numChecks = (\secpar + (n-k+1)\log(\secpar))^2$, we have
    \begin{align}
        1 - (n-k+1)\frac{\log(\secpar)}{\secpar + (n-k+1)\log(\secpar)} 
        &> 1 - \frac{(n-k+1)\log(\secpar)}{(n-k+1)\log(\secpar)} 
        = 0
    \end{align}
    Finally, observe that the choice of parameters in the construction satisfies these constraints.
\ifsubmission\qed\fi\end{proof}

\ifsubmission
\else
\subsection{Proof of Security}

Recall that $c_{\corctr}$ is the $\corctr$'th share to be corrupted (i.e. added to $C$) and $d_{\delctr}$ is the $\delctr$'th share to be deleted (i.e. added to $D$).
Observe that if $c_{k-1+\delctr}$ is corrupted before $d_{\delctr}$ is deleted, then $C\backslash D$ has size $\geq k$ and is authorized, so $\ACD_{(k,n)}$ would abort.

\begin{lemma}\label{claim:strongcd}
    The construction in \Cref{fig:SS-CD-Construction} using parameters from \Cref{fig:acd-parameters-loose} satisfies adaptive certified deletion for threshold secret sharing.
\end{lemma}


\noindent We begin by defining a projector which will be useful for reasoning about how many data indices were \emph{not} destroyed when an adversary produces a valid certificate for a share $i$.
A certificate $\cert_i$ for share $i$ can be parsed as $t$ elements $\cert_{i,1}, \dots, \cert_{t,i}$ of $\bbK$. Denote $\cert_i' = (\cert_{i,j})_{j\in J_{i}}$ to be the subtuple of elements belonging to data indices. For any certificate $\cert$, we define the projector\footnote{This projector defines the ``deletion predicate'' mentioned in the technical overview (\Cref{sec:overview-proving-adaptive}).} 
\[
    \Pi_{\cert} = \sum_{\vect{u} \in \bbK^{t'}:h_{\bbK}(\vect{u})<\ell/2} H^{\otimes t' \lceil \log_{2}(n + 1)\rceil}\ket{\vect{u}+\vect{\cert}'}\bra{\vect{u}+\vect{\cert}'}H^{\otimes t' \lceil \log_{2}(n + 1)\rceil}
\]

Note that $H$ is the Hadamard gate, i.e. it implements a quantum Fourier transform over the binary field $\bbF_2$, and that the Hamming weight is taken over $\bbK$. 

Let $\Adv$ be any adversary which is initialized with a state $\ket{\psi}$ on register $\calR$. For $s\in \{s_0, s_1\}$, define the following $n-k+3$ hybrid experiments, where $\calH_0(s)$ is the original $\ACD(1^\secpar, \ket{\psi}, \Adv, s)$ experiment.
\\

\noindent\underline{$\calH_1(s)$} \\

In $\calH_1(s)$, we sample the shares lazily using polynomial interpolation.

\begin{enumerate}
    \item For each share $i$, sample the set of data indices $J_i\subset [t]$. Then for every share $i$ and every check position $j\in [t]\backslash J_i$, sample the check position $\ket{\psi_{i,j}}$ as in $\calH_0(s)$.
    
    \item For each share $i$, divide the data indices $J_i$ into a left set $J_{i}^L$ of size $\ell$ and a right set $J_{i}^R$ of size $t' - \ell$.
    For each $j\in J_{i}^L$, sample $f(it+j)\gets \bbK$ uniformly at random.

    \item Until $k-1$ shares are corrupted, i.e. $|C| = k-1$, run $\Adv(\{S_j\}_{j\in C}, \calR)$ as in $\ACD$, with the following exception. Whenever $\Adv$ corrupts a new share by outputting $(c_{\corctr}, \calR)$, finish preparing share $c_{\corctr}$ by sampling $f(c_{\corctr}t + j)\gets \bbK$ uniformly at random for every $j \in J_{c_i}^{R}$. 
    
    At the end of this step, exactly $p = (k-1)t' + (n-k+1)\ell$ points of $f$ have been determined, in addition to $f(0) = s$. This uniquely determines $f$.

    \item Continue to run $\Adv(\{S_j\}_{j\in C}, \calR)$ as in $\ACD$, with the following exception. Whenever $\Adv$ corrupts a new share by outputting $(c_{k-1+\delctr}, \calR)$, finish preparing $c_{k-1+\delctr}$ by interpolating the points in $J_{c_{k-1+\delctr}}^{R}$ using share $d_{\delctr}$ and any other set of $p-t'$ points that have already been determined on $f$.
%
    Specifically, let 
    \[
    \mathsf{Int}_{k-1+\delctr} \subset \{0\} \cup \bigcup_{m\in C} \{mt+j:j\in J_m\} \cup \bigcup_{m\notin C} \{mt+j:j\in J_{m}^{L}\}
    \] 
    
    be any set of $p+1$ indices to be used in the interpolation, such that 
    \[\{d_{\delctr} t + j:j \in J_{d_{\delctr}}\} \subset \mathsf{Int}_{k-1+\delctr}\]
    For each $j\in J_{c_{k-1+\delctr}}^{R}$, compute 
    \[
        f(c_{k-1+\delctr} t+j) \gets \Interpolate_{p}\left(c_{k-1+\delctr} t+j, \{(m, f(m)): m\in \mathsf{Int}_{k-1+\delctr}\}\right)
    \]
    See \Cref{sec:prelim-poly-RS} for the definition of $\Interpolate$.

    Note that if $\ACD$ does not abort in a round, $|C\backslash D| \leq k-1$. In the round where  $\Adv$ corrupts $c_{k-1+i}$, $|C| = k-1+i$, so $d_i$ has already been determined. 
\end{enumerate}

\noindent\underline{$\calH_2(s)$} 
\\

In $\calH_2(s)$, we purify the share sampling using a register $\chalReg = (\chalReg_1, \dots, \chalReg_n)$ which is held by the challenger. The challenger will maintain a copy of share $i$ in register $\chalReg_{i} = (\chalReg_{i,1}, \dots, \chalReg_{i,t})$. Both $\calS$ and $\chalReg$ are initialized to $\ket{0}$ at the beginning of the experiment.

\begin{enumerate}
    \item For each share $i$, sample the set of data indices $J_i\subset [t]$. Then for every share $i$ and every check position $j\in [t]\backslash J_i$, prepare the state
    \[
        \propto \sum_{y\in\bbK} \ket{y}^{\calS_{i,j}} \otimes \ket{y}^{\chalReg_{i,j}}
    \]
    Measure $\chalReg_{i,j}$ in the Hadamard basis to obtain $y_{i,j}$ for the verification key.
    
    \item Divide each $J_i$ into $J_{i}^L$ and $J_{i}^R$ as in $\calH_1(s)$. For each $j\in J_{i}^L$, prepare the state 
    \[
        \propto \sum_{y\in\bbK} \ket{y}^{\calS_{i,j}} \otimes \ket{y}^{\chalReg_{i,j}}
    \]

    \item Until $k-1$ shares are corrupted, i.e. $|C| = k-1$, run $\Adv(\{S_j\}_{j\in C}, \calR)$ as in $\ACD$, with the following exception. Whenever $\Adv$ corrupts a new share by outputting $(c_{\corctr}, \calR)$, for every $j\in J_{c_{\corctr}}^R$ prepare the state 
    \[
        \propto \sum_{y\in\bbK} \ket{y}^{\calS_{c_{\corctr},j}} \otimes \ket{y}^{\chalReg_{c_{\corctr},j}}
    \]

    \item Continue to run $\Adv(\{S_j\}_{j\in C}, \calR)$ as in $\ACD$, with the following exception whenever $\Adv$ corrupts a new share by outputting $(c_{k-1+\delctr}, \calR)$. Let $\mathsf{Int}_{k-1+\delctr}$ be the set of indices to be used in interpolation for share $c_{k-1+\delctr}$, as in $\calH_1(s)$. For each $j\in J_{c_{k-1+\delctr}}$, compute
    \[
        \chalReg_{c_{k-1+\delctr},j} \gets \Interpolate_{p}\left(c_{k-1+\delctr}t + j, \left(mt+j, \calS_{m, j} \right)_{mt + j \in \mathsf{Int}_{k-1+\delctr}} \right)
    \]
    Finally, copy $\chalReg_{c_{k-1+\delctr},j}$ into $\calS_{c_{k-1+\delctr},j}$ in the computational basis, i.e. perform a controlled NOT operation with source register $\chalReg_{c_{k-1+\delctr},j}$ and target register $\calS_{c_{k-1+\delctr},j}$.
    
\end{enumerate}

We emphasize that the timing of initializing each $\calS_{i,j}$ is the same as in $\calH_1(s)$. Note that since $\calH_2(s)$ outputs either $\bot$ or $\Adv$'s view, register $\calC$ never appears in the output of the experiment.
\\

\noindent\underline{$\calH_{2+i}(s)$ for $i\in [n-k+1]$}\\

The only difference between $\calH_{2+i}$ and $\calH_{3+i}$ is that when the $i$'th share $d_i$ is deleted in $\calH_{3+i}$ (i.e. $D$ reaches size $i$), the challenger performs a ``deletion predicate'' measurement on register $\chalReg_{d_i}$. Specifically, let $\cert_{d_i}$ be the certificate output by $\Adv$ for share $d_i$. Immediately after verifying $\cert_{d_i}$ and adding $d_i$ to $D$, the challenger measures the data positions in register $\chalReg_{d_{i}}$ (i.e. register $(\chalReg_{d_{i}, j})_{j\in J_{d_{i}}}$) with respect to the binary projective measurement $\{\Pi_{\cert_{k + i}}, I - \Pi_{\cert_{k + i}}\}$. If the measurement result is ``reject'' (i.e. $I - \Pi_{\cert_{k  + i}}$), immediately output $\bot$ in the experiment.
The difference between $\calH_2$ and $\calH_3$ is the same, for $i=1$.
\\

In addition to hybrids $\calH_0$ through $\calH_{3+n-k}$, we define a set of simulated experiments.  Each $\Sim_i$ will be useful for reasoning about hybrid $\calH_{2+i}$. $\Sim_i$ is similar to $\calH_{2+i}$ except that all of the shares are randomized, whereas in $\calH_{2+i}$, shares corrupted after $c_{k-1+i}$ are interpolated.
\\

\noindent\underline{$\Sim_i$ for $i\in [n-k+1]$} \\

Run the $\ACD(1^\secpar, \ket{\psi}, \Adv, s)$ experiment, with the following exceptions. 
\begin{itemize}
    \item Do \emph{not} initialize $(\calS_1, \dots, \calS_n, \verkey) \gets \Split_{\bbS}(1^\secpar, s)$ in step 1.
    \item Whenever $\Adv$ corrupts a new share by outputting $(c_{\corctr}, \calR)$, prepare the state 
    \[
        \propto \sum_{\vect{y}\in\bbK^t} \ket{\vect{y}}^{\calS_{c_\corctr}} \otimes \ket{\vect{y}}^{\chalReg_{c_{\corctr}}}
    \]
    Then, sample the set of data indices $J_{c_\corctr}\subset [t]$ of size $t'$ and for each check index $j\in [t]\backslash J_{c_\corctr}$ measure $\chalReg_{c_{\corctr},j}$ in the Hadamard basis to obtain $y_{\corctr,j}$ for the verification key.
    
    \item For the first $i$ deletions $d_\delctr$ where $\delctr\leq i$, immediately after the challenger verifies $\cert_{d_\delctr}$ and adds $d_\delctr$ to $D$, it measures the data positions in register $\chalReg_{d_{\delctr}}$ with respect to the binary projective measurement $\{\Pi_{\cert_{d_\delctr}}, I - \Pi_{\cert_{d_\delctr}}\}$. If the measurement result is ``reject'', immediately output $\bot$ in the experiment.
\end{itemize}

\begin{claim}\label{claim:h0-h2} 
For every secret $s$,
    \[
        \tracedist[\calH_0(s), \calH_2(s)] = 0
    \]
\end{claim}
\begin{proof}
    It is sufficient to show that $\tracedist[\calH_0(s), \calH_1(s)] = 0$ and $\tracedist[\calH_1(s), \calH_2(s)] = 0$. The former is true by the correctness of polynomial interpolation (see \Cref{sec:prelim-poly-RS}). To see the latter, observe that steps 1, 2, and 3 in $\calH_2(s)$ are equivalent to sampling a uniformly random state (in \emph{any} basis) in register $\calS_{i,j}$ by preparing a uniform superposition over the basis elements in $\calS_{i,j}$, then performing a delayed measurement from $\calS_{i,j}$ to $\chalReg_{i,j}$ in that basis. Observe that steps 1, 2, and 3 in $\calH_{1}(s)$ also sample uniformly random states in $\calS_{i,j}$. Now consider step 4. In $\calH_2(s)$, step 4 performs a (classical) polynomial interpolation using copies of points $(it + j, f(it + j))$ that are obtained by measuring $\calS_{i,j}$. This is equivalent to directly interpolating using $\calS_{i,j}$ if $\calS_{i,j}$ contained a computational basis state, which is the case in $\calH_1(s)$.
\ifsubmission\qed\fi\end{proof}

We show that $\calH_2$ has negligible trace distance from $\calH_{3 + n - k}$ in \Cref{claim:h2pi-to-h3pi}. To prove \Cref{claim:h2pi-to-h3pi}, we will need an additional fact which we show in \Cref{claim:truncated-is-sim}. \Cref{claim:truncated-is-sim} will also show that the final hybrid $\calH_{3 + n - k}$ has zero trace distance from $\Sim_{n-k+1}$, which is independent of the secret $s$. 

Let $\calH_i[c_\corctr](s)$ denote the truncated game where $\calH_i(s)$ is run until the end of the round where the $\corctr$'th corruption occurs, i.e. when $|C|$ reaches $\corctr$. At this point, $\calH_i[c_\corctr](s)$ outputs the adversary's register $\calR$ and the set of corrupted registers $\{\calS_{j}\}_{j\in C}$, unless the game has ended earlier (e.g. from an abort).\footnote{The truncated version of the game outputs both the set of corrupted registers and $\calR$, while the full version only outputs $\calR$. In the full version, the adversary can move whatever information it wants into $\calR$. However, the truncated game ends early, so the adversary may not have done this when the game ends. Outputting the corrupted registers directly ensures that they appear in the output in some form if the game does not abort.}
Let $\calH_i[d_\delctr](s)$ similarly represent the truncated game where $\calH_i(s)$ is run until the end of the round where the $\delctr$'th deletion occurs, i.e. when $|D|$ reaches $\delctr$.
Define $\Sim_i[c_\corctr]$ and $\Sim_i[d_{\delctr}]$ similarly. 

Observe that after the $n$'th corruption in any hybrid experiment, the rest of the challenger's actions in the experiment is independent of the secret $s$. Therefore for every hybrid $\calH_i$ and every pair of secrets $(s_0, s_1)$,

\[
\tracedist[\calH_i[c_n](s_0), \calH_i[c_n](s_1)] = \tracedist[\calH_i(s_0), \calH_i(s_1)]
\]

\begin{claim}\label{claim:truncated-is-sim}
    For every $i \in [0,n-k+1]$ and every secret $s$, 
    \[
        \tracedist[\calH_{2+i}[c_{k-1+i}](s), \Sim_i[c_{k-1+i}]] = 0
    \]
\end{claim}

Combining this claim with the previous observation about the relation of a truncated experiment to its full version, it is clear that 
\begin{align*}
    \tracedist[\calH_{3 + n - k}(s_0), \calH_{3+n-k}(s_1)] 
    &= \tracedist[\calH_{3 + n - k}[c_n](s_0), \calH_{3+n-k}[c_n](s_1)]
    \\
    &= \tracedist[\Sim_{n-k+1}[c_n], \Sim_{n-k+1}[c_n]] 
    \\
    &= 0
\end{align*}
By \Cref{claim:h2pi-to-h3pi}, we have
\[
    \tracedist[\calH_{2}(s_0), \calH_{2}(s_1)] \leq \tracedist[\calH_{3+n-k}(s_0), \calH_{3+n-k}(s_1)] + \negl
\]
Therefore, combining claims \Cref{claim:h0-h2}, \Cref{claim:truncated-is-sim}, and \Cref{claim:h2pi-to-h3pi}, we have
\begin{align*}
    \tracedist[\calH_{0}(s_0), \calH_{0}(s_1)] 
    &\leq \tracedist[\calH_{3+n-k}(s_0), \calH_{3+n-k}(s_1)] + \negl 
    \\
    &\leq 0 + \negl
\end{align*}
which completes the proof. All that remains is to prove \Cref{claim:truncated-is-sim}, and \Cref{claim:h2pi-to-h3pi}.

\begin{proof}[Proof of \Cref{claim:truncated-is-sim}.]
    We proceed via induction. This is clearly true for $i = 0$, since the first $k-1$ shares to be corrupted are prepared as maximally mixed states in both $\calH_{2}(s)$ and $\Sim$.

    Before addressing the case of $i > 0$, we define some notation for our specific application of interpolation. When preparing a share $c_{k-1+i}$ after it is corrupted, the challenger interpolates evaluations of $f$ into a register
    \[
        \chalReg_{c_{k-1+i}}^R \coloneqq (\chalReg_{c_{k-1+i},j})_{j\in J_{c_{k-1+i}}^{R}}
    \]
    $\chalReg_{c_{k-1+i}}^R$ consists of the right data positions of share $c_{k-1+i}$ and contains $t'-\ell$ $\bbK$-qudits. To do the interpolation, the challenger uses evaluations of $f$ contained in registers 
    \[
        \chalReg_{d_{i}}' \coloneqq (\chalReg_{d_i,j})_{j\in J_{d_{i}}}
    \]
    and some other registers which we group as $\calI$. $\chalReg_{d_{i}}'$ consists of the data positions in share $d_i$ and contains $t'$ $\bbK$-qudits.
    Since polynomial interpolation is a linear operation over $\bbK$, the system immediately after $c_{k-1+i}$ is prepared can be described as a state
    \[
        \sum_{\substack{\vect{x_1} \in \bbK^{t'}\\\vect{x_2} \in \bbK^{d - t'} }}  \alpha_{\vect{x_1}, \vect{x_2}} \ket{\vect{x_1}}^{\chalReg_{d_{i}}'}
        \otimes \ket{\vect{x_2}}^{\calI} 
        \otimes \ket{\vect{\vect{R_1}} \vect{x_1} + \vect{R_2} \vect{x_2}}^{\chalReg_{c_{k-1+i}}^R} 
        \otimes \ket{\vect{R_1} \vect{x_1} + \vect{R_2} \vect{x_2}}^{\calS_{c_{k-1+i}}^R} 
        \otimes \ket{\phi_{\vect{x_1}, \vect{x_2}}}^{\chalReg', \calS', \calR}
    \]
    where $\vect{R_1}\in \bbK^{(t' - \delLoss) \times t' }$ and $\vect{R_2} \in \bbK^{(t' - \delLoss) \times (d + 1 - t')}$ are submatrices of the interpolation transformation, where $\calS_{c_{k-1+i}}^R$ contains the copy of the evaluations in $\chalReg_{c_{k-1+i}}^R$, where $\chalReg'$ and $\calS'$ respectively consist of the unmentioned registers of $\chalReg$ and $\calS$, and where $\calR$ is the adversary's internal register.

    Now we will show that the claim holds for $i+1$ if it holds for $i$. Define the following hybrid experiments.
    \begin{itemize}
        \item $\calH_{3+i}[c_{k+i}]$: Recall that the only difference between $\calH_{2+i}$ and $\calH_{3+i}$ is an additional measurement made in the same round that the $(i+1)$'th share is deleted, i.e. when $|D|$ reaches $i+1$.

        \item $\calH_{3+i}'[c_{k+i}]$: The only difference from $\calH_{3+i}[c_{k+i}]$ occurs when the adversary requests to corrupt share $c_{k+i}$. When this occurs, the challenger prepares the right data positions of share $c_{k+i}$ as the state
        \[
            \propto \sum_{\vect{y}\in\bbK^{t'}} \ket{\vect{y}}^{\chalReg_{c_{k+i}}^R} \otimes \ket{\vect{y}}^{\calS_{c_{k+i}}^R}
        \]
        
        \item $\Sim_i[c_{k+i}]$: The only difference from $\calH_{3+i}'[c_{k+i}]$ is that $\Sim_i[c_{k+i}]$ is run until $c_{k+i}$ is corrupted, then the experiment is finished according to $\calH_{3+i}'[c_{k+i}]$.
        
    \end{itemize}

    We first show that $\calH_{3 + i}'[c_{k+i}]$ and $\Sim_i[c_{k+i}]$ are close. By the inductive hypothesis, 
    \[
        \tracedist[\calH_{2+i}[c_{k-1+i}](s), \Sim_{i-1}[c_{k-1+i}]] = 0
    \]
    \noindent Note that $\calH_{2+i}(s)$ and $\calH'_{3+i}(s)$ are identical until the $(i+1)$'th deletion $d_{i+1}$. Similarly, $\Sim_{i-1}$ and $\Sim_{i}$ are identical until the $(i+1)$'th deletion. Therefore
    \[
        \tracedist[\calH'_{3+i}[d_{i}](s), \Sim_{i}[d_{i}]] = 0
    \]
    Finally, observe that if the experiment does not abort, then $d_{i}$ is deleted before $c_{k+i}$ is corrupted (otherwise $|C\backslash D| = k$ during some round). Because of this, $\calH_{3 + i}'[c_{k+i}]$ and $\Sim[c_{k+i}]$ behave identically after the round where $d_{i}$ is corrupted. Therefore
    \[
        \tracedist[\calH_{3 + i}'[c_{k+i}], \Sim[c_{k+i}]] = 0
    \]

    It remains to be shown that 
    \[
        \tracedist[\calH_{3+i}[c_{k+i}], \calH_{3+i}'[c_{k+i}]] = 0
    \]

    The only difference between $\calH_{3+i}[c_{k+i}]$ and $\calH_{3+i}'[c_{k+i}]$ is at the end of the last round, where $c_{k+i}$ is corrupted.\footnote{In the definition of the $\ACD$ experiment, the corruption set $C$ is updated in between adversarial access to the corrupted registers, which occur once at the beginning of each round. The truncated game outputs according to the updated $C$, including $\calS_{c_{k+i}}$.} If the experiment reaches the end of this round without aborting, then $d_i$ has already been corrupted, since otherwise at some point $|C\backslash D| = k$. Furthermore, the deletion predicate measurement on $\chalReg_{d_{i}}$ must have accepted or else the experiment also would have aborted. It is sufficient to prove that the two experiments have $0$ trace distance conditioned on not aborting.

    In $\calH_{3+i}'[c_{k+i}]$, if the experiment does not abort then its end state (after tracing out the challenger's $\chalReg$ register) is
    \[
        \sum_{\substack{\vect{x_1} \in \bbK^{t'}\\\vect{x_2} \in \bbK^{d - t'} \\ x_3 \in \bbK^{t' - \delLoss}}} |\alpha_{\vect{x_1}, \vect{x_2}}|^2 \ket{x_3}\bra{x_3}^{\calS_{c_{k-1+i}}^R} 
        \otimes \Tr^{\chalReg'}\left[\ket{\phi_{\vect{x_1}, \vect{x_2}}}\bra{\phi_{\vect{x_1}, \vect{x_2}}}^{\chalReg', \calS', \calR}\right]
    \]
    


    We now calculate the end state of $\calH_{3+i}'[c_{k+i}]$, conditioned on it not aborting. In this case, the challenger for $\calH_{3+i}'[c_{k+i}]$ prepares the next corrupted register $\calS_{c_{k+i}}$ by a polynomial interpolation which uses registers $\chalReg_{d_{i}}$ and $\calI$. The deletion predicate measurement on $\chalReg_{d_i}$ must have accepted to avoid an abort, so before performing the interpolation, the state of the system is of the form
    \[
        \ket{\gamma}^{A,C',\mathsf{Int}_2,C_{d_{k+i}}} 
        = \sum_{\vect{u} \in \bbK^{t'}:h_{\bbK}(u)<\ell/2} \alpha_{\vect{u}} \ket{\psi_{\vect{u}}}^{\calS,\chalReg',\calI}
        \otimes H^{\otimes t' \lceil \log_{2}(n + 1)\rceil}\ket{\vect{u}+\vect{\cert}_{k+i}}^{\chalReg_{d_{i}}}
    \]
     We will apply \Cref{lem:extractor} with input size $t'$ and output size $t'-\ell$ to show that share $d_{i}$ contributes uniform randomness to the preparation of $c_{k+i}$. Observe that $\chalReg_{d_{i}}$ contains $t'$  $\bbK$-qudits and the interpolation target $\chalReg_{c_{k+i}}^R$ contains $t' - \ell$ $\bbK$-qudits. Furthermore,  $[R_1 | R_2] \in \bbK^{(t' - \delLoss) \times (d + 1)}$ is an interpolation matrix for a polynomial of degree $p$. By \Cref{fact:interpolation-linear-independence}, any $t' - \delLoss$ columns of $[R_1 | R_2]$ are linearly independent. In particular, any $t' - \delLoss$ columns of $R_1$ are linearly independent. Finally, note that $(t' - (t'-\ell))/2 = \ell/2$. Therefore by \Cref{lem:extractor}, the state of the system after preparing register $\chalReg_{c_{k+i}}^R$ and tracing out $\chalReg_{d_{i}}$ when the experiment ends at the end of this round is
     \[
        \sum_{\vect{x_3} \in \bbK^{t' - \delLoss}} \Tr^{\chalReg_{d_{i}}} \left[\ket{\gamma_{\vect{x_3}}}\bra{\gamma_{\vect{x_3}}}^{\chalReg,\calS,\calR} \right]
    \]
    where
    \[
        \ket{\gamma_{\vect{x_3}}} = 
        \sum_{\substack{\vect{x_1} \in \bbK^{t'}\\\vect{x_2} \in \bbK^{d - t'} }}  \alpha_{\vect{x_1}, \vect{x_2}} \ket{\vect{x_1}}^{\chalReg_{d_{i}}} 
        \otimes \ket{\vect{x_2}}^{\calI} 
        \otimes \ket{\vect{x_3} + \vect{R_2} \vect{x_2}}^{\chalReg_{c_{k+i}}} 
        \otimes \ket{\vect{x_3} + \vect{R_2} \vect{x_2}}^{\calS_{c_{k+i}}} 
        \otimes \ket{\varphi_{\vect{x_1}, \vect{x_2}}}^{\chalReg', \calS', \calR}
    \]

    After this round, $\calH_{3 + i}[c_{k+i}]$ ends and register $\chalReg$ is traced out. This yields the state
    \begin{align*}
        &\sum_{\vect{x_3} \in \bbK^{t' - \delLoss}}\Tr^{\chalReg}\left[\ket{\gamma_{\vect{x_3}}}\bra{\gamma_{\vect{x_3}}}^{\chalReg, \calS, \calR} \right]
        \\
        &\quad= \sum_{\substack{\vect{x_1} \in \bbK^{t'}\\\vect{x_2} \in \bbK^{d - t'} \\ \vect{x_3} \in \bbK^{t' - \delLoss}}} 
            |\alpha_{\vect{x_1}, \vect{x_2}}|^2 \ket{\vect{x_3} + \vect{R_2}\vect{x_2}}\bra{\vect{x_3} + \vect{R_2}\vect{x_2}}^{\calS_{c_{k+i}}^R} 
            \otimes \Tr^{C'}\left[\ket{\phi_{\vect{x_1}, \vect{x_2}}}\bra{\phi_{\vect{x_1}, \vect{x_2}}}^{\chalReg', \calS', \calR}\right]
        \\
        &\quad= \sum_{\substack{\vect{x_1} \in \bbK^{t'}\\\vect{x_2} \in \bbK^{d - t'} \\ \vect{x_4} \in \bbK^{t' - \delLoss}}} 
            |\alpha_{\vect{x_1}, \vect{x_2}}|^2 \ket{\vect{x_4}}\bra{\vect{x_4}}^{\calS_{c_{k+i}}^R} 
            \otimes \Tr^{\chalReg'}\left[\ket{\phi_{\vect{x_1}, \vect{x_2}}}\bra{\phi_{\vect{x_1}, \vect{x_2}}}^{\chalReg', \calS', \calR}\right]
    \end{align*}
    where $\vect{x_4} = \vect{x_3} + \vect{R_2}\vect{x_2}$. This state is identical to the state at the end of $\calH'_{3 + i}[c_{k+i}]$ conditioned on the experiments not aborting.

\ifsubmission\qed\fi\end{proof}

\begin{claim} \label{claim:h2pi-to-h3pi}
    For every $i\in [0,n-k]$ and every secret $s$,
    \[
    \tracedist[\calH_{2+i}(s), \calH_{3 + i}(s)] = \negl
    \]
\end{claim}
\begin{proof}
    The only difference between $\calH_{2+i}(s)$ and $\calH_{3 + i}(s)$ is an additional deletion predicate measurement $\Pi_{\cert_{d_{i+1}}}$ during the round where $d_{i+1}$ is corrupted. Say the deletion predicate accepts with probability $1-\epsilon$. Then the Gentle Measurement Lemma (\Cref{lem:gentle-measurement}) implies that, conditioned on the deletion predicate accepting, the distance between $\calH_{2+i}(s)$ and $\calH_{3 + i}(s)$ is at most $2\sqrt{\epsilon}$. We upper bound the case where the deletion predicate rejects by $1$ to obtain
    \[
        \tracedist[\calH_{2+i}(s), \calH_{3 + i}(s)] \leq (1-\epsilon)2\sqrt{\epsilon} + \epsilon
    \]
    
    Thus, it is sufficient to show that $\epsilon = \negl$, i.e. the deletion predicate accepts with high probability on $\chalReg_{d_{i+1}}$ in $\calH_{3+i}$.
    To show this, we consider the following hybrids, and claim that the probability that the deletion predicate accepts on $\chalReg_{d_{i+1}}$ is \emph{identical} in each of the hybrids.
    \begin{itemize}
        \item $\calH_{3+i}$
        \item $\calH_{3+i}[d_{i+1}]$: The only difference is that the game ends after the round where $d_{i+1}$ is deleted.
        \item $\Sim_{i+1}[d_{i+1}]$: Recall that the only difference between $\calH_{3+i}$ and $\Sim_{i+1}$ is that every share $j$ is prepared as the maximally entangled state
        \[
            \sum_{\vect{x} \in \bbK^{t}} \ket{\vect{x}}^{C_j} \otimes \ket{\vect{x}}^{A_j}
        \]
        
        \item $\Sim'_{i+1}[d_{i+1}]$: The same as $\Sim_{i+1}[d_{i+1}]$, except that after preparing the maximally entangled state for each share $j$, we delay choosing $J_j$ and measuring the check indices $\chalReg_{i,j}$ for $j\in [t]\backslash J_j$. These are now done immediately after $\Adv$ deletes $j$ by outputting $(\cert_j, j, \calR)$, and before the challenger verifies $\cert_j$.
    \end{itemize}

    Observe that $\calH_{3+i}$ and $\calH_{3+i}[d_{i+1}]$ are identical until the deletion predicate measurement in the round where $d_{i+1}$ is deleted, so the probability of acceptance is identical. By \Cref{claim:truncated-is-sim}, 
    \[
        \tracedist[\calH_{3+i}[c_{k+i}](s),\Sim_{i+1}[c_{k+i}]] = 0
    \]
    Share $d_{i+1}$ is deleted before share $c_{k+i}$ is corrupted in both $\calH_{3+i}(s)$ and $\Sim_{i+1}$, unless they abort. Therefore 
    \[
        \tracedist[\calH_{3+i}[d_{i+1}](s),\Sim_{i+1}[d_{i+1}]] = 0
    \]
    and the probability of acceptance is identical in $\calH_{3+i}[d_{i+1}](s)$ and $\Sim_{i+1}[d_{i+1}]$. Finally, the probability of acceptance is identical in $\Sim'_{i+1}[d_{i+1}]$ because the register $C$ is disjoint from the adversary's registers. 

    Thus, it suffices to show that $\epsilon = \negl$ in $\Sim'_{i+1}[d_{i+1}]$. Since $\bbK$ forms a vector space over $\bbF_2$, the certificate verification measurement and $\Pi_{\cert_{d_{i+1}}}$ are diagonal in the binary Fourier basis (i.e. the Hadamard basis) for every $\cert$. Therefore the probability that $\Verify$ accepts $\cert_{d_{i+1}}$ but the deletion predicate measurement \emph{rejects} $\chalReg_{d_{i+1}}$ is
    \[
        \epsilon = \Pr_{\substack{\vect{\cert}, \vect{y} \in \bbK^{t}\\ J\subset [t]: |J| = t'}}\left[\vect{\cert}_{\overline{J}} = \vect{y}_{\overline{J}} \land \Delta_{\bbK}(\vect{\cert}_{J}, \vect{y}_{J}) \geq \frac{\delLoss}{2}\right]
    \]
    where $J$ is the set of data indices for share $d_{i+1}$, where $\overline{J}$ is the set complement of $J$ (i.e. the set of check indices for share $d_{i+1}$), and where $\Delta_{\bbK}(\vect{\cert}_{J}, \vect{y}_{J}) = h_{\bbK}(\vect{\cert}_{J} - \vect{y}_{J})$ is the Hamming distance of $\cert_{J}$ from $y_{J}$. Here, the probability is taken over the adversary outputting a certificate $\cert$ for $d_{k+i}$, the challenger sampling a set of check indices $\overline{J}$, and the challenger measuring all of register $\chalReg_{d_{i+1}}$ in the Hadamard basis to obtain $\vect{y}\in \bbK^t$.
    
    This value can be upper bounded using Hoeffding's inequality, for any fixed $\cert$ and $\vect{y}$ with $\Delta_{\bbK}(\cert_{J}, \vect{y}_{J}) \geq \delLoss/2$. Note that the probability of acceptance is no greater than if the $\numChecks$ check indices $\overline{J}$ are sampled \emph{with} replacement. In this case, the expected number of check indices which do \emph{not} match is 
    \begin{align}
        \geq \frac{\delLoss\numChecks}{2t} 
        &= \frac{t\log(\secpar)}{\secpar + (n-k+1)\log(\secpar)} \frac{(\secpar + (n-k+1)\log(\secpar))^2}{2t} 
        \\
        &= \frac{\log(\secpar)}{2} (\secpar + (n-k+1)\log(\secpar))
    \end{align}
    Therefore Hoeffding's inequality (\Cref{claim:hoeffding}) implies that
    \begin{align}
        \epsilon 
        &\leq 2\exp\left(
            \frac{-2\left(
                    \frac{\log(\secpar)}{2} (\secpar + (n-k+1)\log(\secpar))
                \right)^2}
                {(\secpar + (n-k+1)\log(\secpar))^2}
        \right) 
        \\
        &= 2\exp\left(-\frac{\log^2(\secpar)}{2}\right) 
        \\
        &= \negl
    \end{align}

    
\ifsubmission\qed\fi\end{proof}
\fi 
\section{Acknowledgements}

We thank Orestis Chardouvelis and Dakshita Khurana for collaboration and insightful discussions during the early stages of this research.
\fi 

\bibliographystyle{alpha}
\bibliography{bib/abbrev3,bib/crypto,bib/extra}

\ifextendedabstract
\else
\appendix
\ifsubmission

\fi
\section{Tighter Parameters for the Threshold Construction}

In this section, we give alternate parameter settings for the construction in \Cref{fig:SS-CD-Construction} that result in slightly smaller share sizes. The parameters are described in \Cref{fig:acd-parameters-tight}. The main difference from \Cref{fig:acd-parameters-loose} is that $r$ is slightly smaller, which also impacts $t$.


\begin{figure}[h]
    \begin{framed}
    \begin{minipage}{\textwidth}
         The construction in \Cref{fig:SS-CD-Construction} uses the following parameters.
         \begin{itemize}
            \item Each share consists of $t$ total $\bbK$-registers, where 
            \[
                t = 
                (k+1)\numChecks \left(1 + \frac{(n-k+1)\log(\secpar)}{\sqrt{\numChecks} - (n-k+1)\log(\secpar)}\right) + 1
            \]
            
            \item A share is divided into $\numChecks$ check indices and $t'=t-r$ data indices, where
            \[
                \numChecks = \secpar + (n-k+1)^2 \log^2(\secpar)
            \]

            \item $\ell$ intuitively represents an upper bound on the amount of information which is not destroyed when an adversary generates a valid deletion certificate for a share.
            \[
                \ell = \delLossValueShort
            \]
            
            See the proof of \Cref{claim:strongcd} for a more precise usage of $\ell$.
            
            \item The secret will be encoded in a polynomial of degree
            \[
            d = (k-1)t' + (n-k+1)\delLoss
            \]
         \end{itemize}
    \end{minipage}
    \end{framed}
    \caption{Alternate Parameters for Secret Sharing with Adaptive Certified Deletion}
    \label{fig:acd-parameters-tight}
\end{figure}

\begin{lemma}\label{claim:recon-correctness-tight}
    The construction in \Cref{fig:SS-CD-Construction} has reconstruction correctness with the parameters in \Cref{fig:acd-parameters-tight}.
\end{lemma}
\begin{proof}
    The set $\{(it + j, y_{i,j})\}_{i\in P', j\in [t]}$ contains $kt$ pairs which were obtained by measuring $k$ shares. As mentioned in \Cref{sec:prelim-poly-RS}, if all but $e<(kt - d)/2$ of these pairs $(it+j, y_{i,j})$ satisfy $y_{i,j} = f(it + j)$, then $\mathsf{Correct}_{\bbK,d}$ recovers the original polynomial $f$, where $f(0) = s$. The only points which do not satisfy this are the check positions, of which there are $\numChecks$ per share, for a total of $k\numChecks$. Therefore for correctness, we require that
    
    \begin{align}
    2k\numChecks &< kt - d
        \\
        &= kt - (k-1)(t-\numChecks) - (n-k+1)\delLoss
        \\
        &= t + (k-1)\numChecks- (n-k+1)\delLoss
    \end{align}
    Therefore $t - (n-k+1)\ell > (k+1)\numChecks$. Substituting $\delLoss = \delLossValueShort$ yields
    \begin{align}
        t\left(1 - (n-k+1)\frac{\log(\secpar)}{\sqrt{r}}\right) &> (k+1)\numChecks
        \\
        t &> (k+1)\numChecks\frac{1}{1 - (n-k+1)\frac{\log(\secpar)}{\sqrt{r}}}
        \label{eq:correctness-nonnegative-tight}
        \\
        &= (k+1)\numChecks\frac{\sqrt{\numChecks}}{\sqrt{\numChecks} - (n-k+1)\log(\secpar)}
        \\
        &= (k+1)\numChecks \left(1 + \frac{(n-k+1)\log(\secpar)}{\sqrt{\numChecks} - (n-k+1)\log(\secpar)}\right)
    \end{align}
    Note that \Cref{eq:correctness-nonnegative-tight} requires that $\left(1 - (n-k+1)\delLossValueShort\right) > 0$. Since the number of check positions is $\numChecks = \secpar + (n-k+1)^2\log^2(\secpar)$, we have
    \begin{align}
        1 - (n-k+1)\frac{\log(\secpar)}{\sqrt{\secpar + (n-k+1)^2\log^2(\secpar)}} 
        &> 1 - \frac{(n-k+1)\log(\secpar)}{(n-k+1)\log(\secpar)} = 0
    \end{align}
    Finally, observe that the choice of parameters in the construction satisfies these constraints.
\ifsubmission\qed\fi\end{proof}

\begin{lemma}
    The construction in \Cref{fig:SS-CD-Construction} has adaptive certified deletion security with the parameters in \Cref{fig:acd-parameters-tight}.
\end{lemma}
\begin{proof}[Proof Sketch.]
    The proof is almost the same as that of \Cref{claim:strongcd}, except for the application of Hoeffding's inequality in \Cref{claim:h2pi-to-h3pi}. The expected number of check indices which do not match becomes
     \begin{align*}
        \geq \frac{\delLoss\numChecks}{2t} 
        &= \frac{t\log(\secpar)}{\sqrt{\secpar + (n-k+1)^2\log^2(\secpar)}} \frac{\secpar + (n-k+1)^2\log^2(\secpar)}{2t} 
        \\
        &= \frac{\log(\secpar)}{2} \sqrt{\secpar + (n-k+1)^2\log^2(\secpar)}
    \end{align*}
    Then Hoeffding's inequality implies
    \begin{align*}
        \epsilon 
        &\leq 2\exp\left(
            \frac{-2\left(\frac{\log(\secpar)}{2} \sqrt{\secpar + (n-k+1)^2\log^2(\secpar)}\right)^2}{\secpar + (n-k+1)^2\log^2(\secpar)}
        \right) 
        \\
        &= 2\exp\left(-\frac{\log^2(\secpar)}{2}\right) 
        \\
        &= \negl
    \end{align*}
\ifsubmission\smallskip\qed\fi\end{proof}

\ifsubmission
\section{Proof of Adaptive Certified Deletion Security}\label{app:acd-proof}

Recall that $c_{\corctr}$ is the $\corctr$'th share to be corrupted (i.e. added to $C$) and $d_{\delctr}$ is the $\delctr$'th share to be deleted (i.e. added to $D$).
Observe that if $c_{k-1+\delctr}$ is corrupted before $d_{\delctr}$ is deleted, then $C\backslash D$ has size $\geq k$ and is authorized, so $\ACD_{(k,n)}$ would abort.

\begin{lemma}\label{claim:strongcd}
    The construction in \Cref{fig:SS-CD-Construction} using parameters from \Cref{fig:acd-parameters-loose} satisfies adaptive certified deletion for threshold secret sharing.
\end{lemma}


\noindent We begin by defining a projector which will be useful for reasoning about how many data indices were \emph{not} destroyed when an adversary produces a valid certificate for a share $i$.
A certificate $\cert_i$ for share $i$ can be parsed as $t$ elements $\cert_{i,1}, \dots, \cert_{t,i}$ of $\bbK$. Denote $\cert_i' = (\cert_{i,j})_{j\in J_{i}}$ to be the subtuple of elements belonging to data indices. For any certificate $\cert$, we define the projector\footnote{This projector defines the ``deletion predicate'' mentioned in the technical overview (\Cref{sec:overview-proving-adaptive}).} 
\[
    \Pi_{\cert} = \sum_{\vect{u} \in \bbK^{t'}:h_{\bbK}(\vect{u})<\ell/2} H^{\otimes t' \lceil \log_{2}(n + 1)\rceil}\ket{\vect{u}+\vect{\cert}'}\bra{\vect{u}+\vect{\cert}'}H^{\otimes t' \lceil \log_{2}(n + 1)\rceil}
\]

Note that $H$ is the Hadamard gate, i.e. it implements a quantum Fourier transform over the binary field $\bbF_2$, and that the Hamming weight is taken over $\bbK$. 

Let $\Adv$ be any adversary which is initialized with a state $\ket{\psi}$ on register $\calR$. For $s\in \{s_0, s_1\}$, define the following $n-k+3$ hybrid experiments, where $\calH_0(s)$ is the original $\ACD(1^\secpar, \ket{\psi}, \Adv, s)$ experiment.
\\

\noindent\underline{$\calH_1(s)$} \\

In $\calH_1(s)$, we sample the shares lazily using polynomial interpolation.

\begin{enumerate}
    \item For each share $i$, sample the set of data indices $J_i\subset [t]$. Then for every share $i$ and every check position $j\in [t]\backslash J_i$, sample the check position $\ket{\psi_{i,j}}$ as in $\calH_0(s)$.
    
    \item For each share $i$, divide the data indices $J_i$ into a left set $J_{i}^L$ of size $\ell$ and a right set $J_{i}^R$ of size $t' - \ell$.
    For each $j\in J_{i}^L$, sample $f(it+j)\gets \bbK$ uniformly at random.

    \item Until $k-1$ shares are corrupted, i.e. $|C| = k-1$, run $\Adv(\{S_j\}_{j\in C}, \calR)$ as in $\ACD$, with the following exception. Whenever $\Adv$ corrupts a new share by outputting $(c_{\corctr}, \calR)$, finish preparing share $c_{\corctr}$ by sampling $f(c_{\corctr}t + j)\gets \bbK$ uniformly at random for every $j \in J_{c_i}^{R}$. 
    
    At the end of this step, exactly $p = (k-1)t' + (n-k+1)\ell$ points of $f$ have been determined, in addition to $f(0) = s$. This uniquely determines $f$.

    \item Continue to run $\Adv(\{S_j\}_{j\in C}, \calR)$ as in $\ACD$, with the following exception. Whenever $\Adv$ corrupts a new share by outputting $(c_{k-1+\delctr}, \calR)$, finish preparing $c_{k-1+\delctr}$ by interpolating the points in $J_{c_{k-1+\delctr}}^{R}$ using share $d_{\delctr}$ and any other set of $p-t'$ points that have already been determined on $f$.
%
    Specifically, let 
    \[
    \mathsf{Int}_{k-1+\delctr} \subset \{0\} \cup \bigcup_{m\in C} \{mt+j:j\in J_m\} \cup \bigcup_{m\notin C} \{mt+j:j\in J_{m}^{L}\}
    \] 
    
    be any set of $p+1$ indices to be used in the interpolation, such that 
    \[\{d_{\delctr} t + j:j \in J_{d_{\delctr}}\} \subset \mathsf{Int}_{k-1+\delctr}\]
    For each $j\in J_{c_{k-1+\delctr}}^{R}$, compute 
    \[
        f(c_{k-1+\delctr} t+j) \gets \Interpolate_{p}\left(c_{k-1+\delctr} t+j, \{(m, f(m)): m\in \mathsf{Int}_{k-1+\delctr}\}\right)
    \]
    See \Cref{sec:prelim-poly-RS} for the definition of $\Interpolate$.

    Note that if $\ACD$ does not abort in a round, $|C\backslash D| \leq k-1$. In the round where  $\Adv$ corrupts $c_{k-1+i}$, $|C| = k-1+i$, so $d_i$ has already been determined. 
\end{enumerate}

\noindent\underline{$\calH_2(s)$} 
\\

In $\calH_2(s)$, we purify the share sampling using a register $\chalReg = (\chalReg_1, \dots, \chalReg_n)$ which is held by the challenger. The challenger will maintain a copy of share $i$ in register $\chalReg_{i} = (\chalReg_{i,1}, \dots, \chalReg_{i,t})$. Both $\calS$ and $\chalReg$ are initialized to $\ket{0}$ at the beginning of the experiment.

\begin{enumerate}
    \item For each share $i$, sample the set of data indices $J_i\subset [t]$. Then for every share $i$ and every check position $j\in [t]\backslash J_i$, prepare the state
    \[
        \propto \sum_{y\in\bbK} \ket{y}^{\calS_{i,j}} \otimes \ket{y}^{\chalReg_{i,j}}
    \]
    Measure $\chalReg_{i,j}$ in the Hadamard basis to obtain $y_{i,j}$ for the verification key.
    
    \item Divide each $J_i$ into $J_{i}^L$ and $J_{i}^R$ as in $\calH_1(s)$. For each $j\in J_{i}^L$, prepare the state 
    \[
        \propto \sum_{y\in\bbK} \ket{y}^{\calS_{i,j}} \otimes \ket{y}^{\chalReg_{i,j}}
    \]

    \item Until $k-1$ shares are corrupted, i.e. $|C| = k-1$, run $\Adv(\{S_j\}_{j\in C}, \calR)$ as in $\ACD$, with the following exception. Whenever $\Adv$ corrupts a new share by outputting $(c_{\corctr}, \calR)$, for every $j\in J_{c_{\corctr}}^R$ prepare the state 
    \[
        \propto \sum_{y\in\bbK} \ket{y}^{\calS_{c_{\corctr},j}} \otimes \ket{y}^{\chalReg_{c_{\corctr},j}}
    \]

    \item Continue to run $\Adv(\{S_j\}_{j\in C}, \calR)$ as in $\ACD$, with the following exception whenever $\Adv$ corrupts a new share by outputting $(c_{k-1+\delctr}, \calR)$. Let $\mathsf{Int}_{k-1+\delctr}$ be the set of indices to be used in interpolation for share $c_{k-1+\delctr}$, as in $\calH_1(s)$. For each $j\in J_{c_{k-1+\delctr}}$, compute
    \[
        \chalReg_{c_{k-1+\delctr},j} \gets \Interpolate_{p}\left(c_{k-1+\delctr}t + j, \left(mt+j, \calS_{m, j} \right)_{mt + j \in \mathsf{Int}_{k-1+\delctr}} \right)
    \]
    Finally, copy $\chalReg_{c_{k-1+\delctr},j}$ into $\calS_{c_{k-1+\delctr},j}$ in the computational basis, i.e. perform a controlled NOT operation with source register $\chalReg_{c_{k-1+\delctr},j}$ and target register $\calS_{c_{k-1+\delctr},j}$.
    
\end{enumerate}

We emphasize that the timing of initializing each $\calS_{i,j}$ is the same as in $\calH_1(s)$. Note that since $\calH_2(s)$ outputs either $\bot$ or $\Adv$'s view, register $\calC$ never appears in the output of the experiment.
\\

\noindent\underline{$\calH_{2+i}(s)$ for $i\in [n-k+1]$}\\

The only difference between $\calH_{2+i}$ and $\calH_{3+i}$ is that when the $i$'th share $d_i$ is deleted in $\calH_{3+i}$ (i.e. $D$ reaches size $i$), the challenger performs a ``deletion predicate'' measurement on register $\chalReg_{d_i}$. Specifically, let $\cert_{d_i}$ be the certificate output by $\Adv$ for share $d_i$. Immediately after verifying $\cert_{d_i}$ and adding $d_i$ to $D$, the challenger measures the data positions in register $\chalReg_{d_{i}}$ (i.e. register $(\chalReg_{d_{i}, j})_{j\in J_{d_{i}}}$) with respect to the binary projective measurement $\{\Pi_{\cert_{k + i}}, I - \Pi_{\cert_{k + i}}\}$. If the measurement result is ``reject'' (i.e. $I - \Pi_{\cert_{k  + i}}$), immediately output $\bot$ in the experiment.
The difference between $\calH_2$ and $\calH_3$ is the same, for $i=1$.
\\

In addition to hybrids $\calH_0$ through $\calH_{3+n-k}$, we define a set of simulated experiments.  Each $\Sim_i$ will be useful for reasoning about hybrid $\calH_{2+i}$. $\Sim_i$ is similar to $\calH_{2+i}$ except that all of the shares are randomized, whereas in $\calH_{2+i}$, shares corrupted after $c_{k-1+i}$ are interpolated.
\\

\noindent\underline{$\Sim_i$ for $i\in [n-k+1]$} \\

Run the $\ACD(1^\secpar, \ket{\psi}, \Adv, s)$ experiment, with the following exceptions. 
\begin{itemize}
    \item Do \emph{not} initialize $(\calS_1, \dots, \calS_n, \verkey) \gets \Split_{\bbS}(1^\secpar, s)$ in step 1.
    \item Whenever $\Adv$ corrupts a new share by outputting $(c_{\corctr}, \calR)$, prepare the state 
    \[
        \propto \sum_{\vect{y}\in\bbK^t} \ket{\vect{y}}^{\calS_{c_\corctr}} \otimes \ket{\vect{y}}^{\chalReg_{c_{\corctr}}}
    \]
    Then, sample the set of data indices $J_{c_\corctr}\subset [t]$ of size $t'$ and for each check index $j\in [t]\backslash J_{c_\corctr}$ measure $\chalReg_{c_{\corctr},j}$ in the Hadamard basis to obtain $y_{\corctr,j}$ for the verification key.
    
    \item For the first $i$ deletions $d_\delctr$ where $\delctr\leq i$, immediately after the challenger verifies $\cert_{d_\delctr}$ and adds $d_\delctr$ to $D$, it measures the data positions in register $\chalReg_{d_{\delctr}}$ with respect to the binary projective measurement $\{\Pi_{\cert_{d_\delctr}}, I - \Pi_{\cert_{d_\delctr}}\}$. If the measurement result is ``reject'', immediately output $\bot$ in the experiment.
\end{itemize}

\begin{proposition}\label{claim:h0-h2} 
For every secret $s$,
    \[
        \tracedist[\calH_0(s), \calH_2(s)] = 0
    \]
\end{proposition}
\begin{proof}
    It is sufficient to show that $\tracedist[\calH_0(s), \calH_1(s)] = 0$ and $\tracedist[\calH_1(s), \calH_2(s)] = 0$. The former is true by the correctness of polynomial interpolation (see \Cref{sec:prelim-poly-RS}). To see the latter, observe that steps 1, 2, and 3 in $\calH_2(s)$ are equivalent to sampling a uniformly random state (in \emph{any} basis) in register $\calS_{i,j}$ by preparing a uniform superposition over the basis elements in $\calS_{i,j}$, then performing a delayed measurement from $\calS_{i,j}$ to $\chalReg_{i,j}$ in that basis. Observe that steps 1, 2, and 3 in $\calH_{1}(s)$ also sample uniformly random states in $\calS_{i,j}$. Now consider step 4. In $\calH_2(s)$, step 4 performs a (classical) polynomial interpolation using copies of points $(it + j, f(it + j))$ that are obtained by measuring $\calS_{i,j}$. This is equivalent to directly interpolating using $\calS_{i,j}$ if $\calS_{i,j}$ contained a computational basis state, which is the case in $\calH_1(s)$.
\ifsubmission\qed\fi\end{proof}

We show that $\calH_2$ has negligible trace distance from $\calH_{3 + n - k}$ in \Cref{claim:h2pi-to-h3pi}. To prove \Cref{claim:h2pi-to-h3pi}, we will need an additional fact which we show in \Cref{claim:truncated-is-sim}. \Cref{claim:truncated-is-sim} will also show that the final hybrid $\calH_{3 + n - k}$ has zero trace distance from $\Sim_{n-k+1}$, which is independent of the secret $s$. 

Let $\calH_i[c_\corctr](s)$ denote the truncated game where $\calH_i(s)$ is run until the end of the round where the $\corctr$'th corruption occurs, i.e. when $|C|$ reaches $\corctr$. At this point, $\calH_i[c_\corctr](s)$ outputs the adversary's register $\calR$ and the set of corrupted registers $\{\calS_{j}\}_{j\in C}$, unless the game has ended earlier (e.g. from an abort).\footnote{The truncated version of the game outputs both the set of corrupted registers and $\calR$, while the full version only outputs $\calR$. In the full version, the adversary can move whatever information it wants into $\calR$. However, the truncated game ends early, so the adversary may not have done this when the game ends. Outputting the corrupted registers directly ensures that they appear in the output in some form if the game does not abort.}
Let $\calH_i[d_\delctr](s)$ similarly represent the truncated game where $\calH_i(s)$ is run until the end of the round where the $\delctr$'th deletion occurs, i.e. when $|D|$ reaches $\delctr$.
Define $\Sim_i[c_\corctr]$ and $\Sim_i[d_{\delctr}]$ similarly. 

Observe that after the $n$'th corruption in any hybrid experiment, the rest of the challenger's actions in the experiment is independent of the secret $s$. Therefore for every hybrid $\calH_i$ and every pair of secrets $(s_0, s_1)$,

\[
\tracedist[\calH_i[c_n](s_0), \calH_i[c_n](s_1)] = \tracedist[\calH_i(s_0), \calH_i(s_1)]
\]

\begin{proposition}\label{claim:truncated-is-sim}
    For every $i \in [0,n-k+1]$ and every secret $s$, 
    \[
        \tracedist[\calH_{2+i}[c_{k-1+i}](s), \Sim_i[c_{k-1+i}]] = 0
    \]
\end{proposition}

Combining this claim with the previous observation about the relation of a truncated experiment to its full version, it is clear that 
\begin{align*}
    \tracedist[\calH_{3 + n - k}(s_0), \calH_{3+n-k}(s_1)] 
    &= \tracedist[\calH_{3 + n - k}[c_n](s_0), \calH_{3+n-k}[c_n](s_1)]
    \\
    &= \tracedist[\Sim_{n-k+1}[c_n], \Sim_{n-k+1}[c_n]] 
    \\
    &= 0
\end{align*}
By \Cref{claim:h2pi-to-h3pi}, we have
\[
    \tracedist[\calH_{2}(s_0), \calH_{2}(s_1)] \leq \tracedist[\calH_{3+n-k}(s_0), \calH_{3+n-k}(s_1)] + \negl
\]
Therefore, combining \Cref{claim:h0-h2}, \Cref{claim:truncated-is-sim}, and \Cref{claim:h2pi-to-h3pi}, we have
\begin{align*}
    \tracedist[\calH_{0}(s_0), \calH_{0}(s_1)] 
    &\leq \tracedist[\calH_{3+n-k}(s_0), \calH_{3+n-k}(s_1)] + \negl 
    \\
    &\leq 0 + \negl
\end{align*}
which completes the proof. All that remains is to prove \Cref{claim:truncated-is-sim}, and \Cref{claim:h2pi-to-h3pi}.

\begin{proof}[\Cref{claim:truncated-is-sim}.]
    We proceed via induction. This is clearly true for $i = 0$, since the first $k-1$ shares to be corrupted are prepared as maximally mixed states in both $\calH_{2}(s)$ and $\Sim$.

    Before addressing the case of $i > 0$, we define some notation for our specific application of interpolation. When preparing a share $c_{k-1+i}$ after it is corrupted, the challenger interpolates evaluations of $f$ into a register
    \[
        \chalReg_{c_{k-1+i}}^R \coloneqq (\chalReg_{c_{k-1+i},j})_{j\in J_{c_{k-1+i}}^{R}}
    \]
    $\chalReg_{c_{k-1+i}}^R$ consists of the right data positions of share $c_{k-1+i}$ and contains $t'-\ell$ $\bbK$-qudits. To do the interpolation, the challenger uses evaluations of $f$ contained in registers 
    \[
        \chalReg_{d_{i}}' \coloneqq (\chalReg_{d_i,j})_{j\in J_{d_{i}}}
    \]
    and some other registers which we group as $\calI$. $\chalReg_{d_{i}}'$ consists of the data positions in share $d_i$ and contains $t'$ $\bbK$-qudits.
    Since polynomial interpolation is a linear operation over $\bbK$, the system immediately after $c_{k-1+i}$ is prepared can be described as a state
    \[
        \sum_{\substack{\vect{x_1} \in \bbK^{t'}\\\vect{x_2} \in \bbK^{d - t'} }}  \alpha_{\vect{x_1}, \vect{x_2}} \ket{\vect{x_1}}^{\chalReg_{d_{i}}'}
        \otimes \ket{\vect{x_2}}^{\calI} 
        \otimes \ket{\vect{\vect{R_1}} \vect{x_1} + \vect{R_2} \vect{x_2}}^{\chalReg_{c_{k-1+i}}^R} 
        \otimes \ket{\vect{R_1} \vect{x_1} + \vect{R_2} \vect{x_2}}^{\calS_{c_{k-1+i}}^R} 
        \otimes \ket{\phi_{\vect{x_1}, \vect{x_2}}}^{\chalReg', \calS', \calR}
    \]
    where $\vect{R_1}\in \bbK^{(t' - \delLoss) \times t' }$ and $\vect{R_2} \in \bbK^{(t' - \delLoss) \times (d + 1 - t')}$ are submatrices of the interpolation transformation, where $\calS_{c_{k-1+i}}^R$ contains the copy of the evaluations in $\chalReg_{c_{k-1+i}}^R$, where $\chalReg'$ and $\calS'$ respectively consist of the unmentioned registers of $\chalReg$ and $\calS$, and where $\calR$ is the adversary's internal register.

    Now we will show that the claim holds for $i+1$ if it holds for $i$. Define the following hybrid experiments.
    \begin{itemize}
        \item $\calH_{3+i}[c_{k+i}]$: Recall that the only difference between $\calH_{2+i}$ and $\calH_{3+i}$ is an additional measurement made in the same round that the $(i+1)$'th share is deleted, i.e. when $|D|$ reaches $i+1$.

        \item $\calH_{3+i}'[c_{k+i}]$: The only difference from $\calH_{3+i}[c_{k+i}]$ occurs when the adversary requests to corrupt share $c_{k+i}$. When this occurs, the challenger prepares the right data positions of share $c_{k+i}$ as the state
        \[
            \propto \sum_{\vect{y}\in\bbK^{t'}} \ket{\vect{y}}^{\chalReg_{c_{k+i}}^R} \otimes \ket{\vect{y}}^{\calS_{c_{k+i}}^R}
        \]
        
        \item $\Sim_i[c_{k+i}]$: The only difference from $\calH_{3+i}'[c_{k+i}]$ is that $\Sim_i[c_{k+i}]$ is run until $c_{k+i}$ is corrupted, then the experiment is finished according to $\calH_{3+i}'[c_{k+i}]$.
        
    \end{itemize}

    We first show that $\calH_{3 + i}'[c_{k+i}]$ and $\Sim_i[c_{k+i}]$ are close. By the inductive hypothesis, 
    \[
        \tracedist[\calH_{2+i}[c_{k-1+i}](s), \Sim_{i-1}[c_{k-1+i}]] = 0
    \]
    \noindent Note that $\calH_{2+i}(s)$ and $\calH'_{3+i}(s)$ are identical until the $(i+1)$'th deletion $d_{i+1}$. Similarly, $\Sim_{i-1}$ and $\Sim_{i}$ are identical until the $(i+1)$'th deletion. Therefore
    \[
        \tracedist[\calH'_{3+i}[d_{i}](s), \Sim_{i}[d_{i}]] = 0
    \]
    Finally, observe that if the experiment does not abort, then $d_{i}$ is deleted before $c_{k+i}$ is corrupted (otherwise $|C\backslash D| = k$ during some round). Because of this, $\calH_{3 + i}'[c_{k+i}]$ and $\Sim[c_{k+i}]$ behave identically after the round where $d_{i}$ is corrupted. Therefore
    \[
        \tracedist[\calH_{3 + i}'[c_{k+i}], \Sim[c_{k+i}]] = 0
    \]

    It remains to be shown that 
    \[
        \tracedist[\calH_{3+i}[c_{k+i}], \calH_{3+i}'[c_{k+i}]] = 0
    \]

    The only difference between $\calH_{3+i}[c_{k+i}]$ and $\calH_{3+i}'[c_{k+i}]$ is at the end of the last round, where $c_{k+i}$ is corrupted.\footnote{In the definition of the $\ACD$ experiment, the corruption set $C$ is updated in between adversarial access to the corrupted registers, which occur once at the beginning of each round. The truncated game outputs according to the updated $C$, including $\calS_{c_{k+i}}$.} If the experiment reaches the end of this round without aborting, then $d_i$ has already been corrupted, since otherwise at some point $|C\backslash D| = k$. Furthermore, the deletion predicate measurement on $\chalReg_{d_{i}}$ must have accepted or else the experiment also would have aborted. It is sufficient to prove that the two experiments have $0$ trace distance conditioned on not aborting.

    In $\calH_{3+i}'[c_{k+i}]$, if the experiment does not abort then its end state (after tracing out the challenger's $\chalReg$ register) is
    \[
        \sum_{\substack{\vect{x_1} \in \bbK^{t'}\\\vect{x_2} \in \bbK^{d - t'} \\ x_3 \in \bbK^{t' - \delLoss}}} |\alpha_{\vect{x_1}, \vect{x_2}}|^2 \ket{x_3}\bra{x_3}^{\calS_{c_{k-1+i}}^R} 
        \otimes \Tr^{\chalReg'}\left[\ket{\phi_{\vect{x_1}, \vect{x_2}}}\bra{\phi_{\vect{x_1}, \vect{x_2}}}^{\chalReg', \calS', \calR}\right]
    \]

    We now calculate the end state of $\calH_{3+i}'[c_{k+i}]$, conditioned on it not aborting. In this case, the challenger for $\calH_{3+i}'[c_{k+i}]$ prepares the next corrupted register $\calS_{c_{k+i}}$ by a polynomial interpolation which uses registers $\chalReg_{d_{i}}$ and $\calI$. The deletion predicate measurement on $\chalReg_{d_i}$ must have accepted to avoid an abort, so before performing the interpolation, the state of the system is of the form
    \[
        \ket{\gamma}^{A,C',\mathsf{Int}_2,C_{d_{k+i}}} 
        = \sum_{\vect{u} \in \bbK^{t'}:h_{\bbK}(u)<\ell/2} \alpha_{\vect{u}} \ket{\psi_{\vect{u}}}^{\calS,\chalReg',\calI}
        \otimes H^{\otimes t' \lceil \log_{2}(n + 1)\rceil}\ket{\vect{u}+\vect{\cert}_{k+i}}^{\chalReg_{d_{i}}}
    \]
     We will apply \Cref{lem:extractor} with input size $t'$ and output size $t'-\ell$ to show that share $d_{i}$ contributes uniform randomness to the preparation of $c_{k+i}$. Observe that $\chalReg_{d_{i}}$ contains $t'$  $\bbK$-qudits and the interpolation target $\chalReg_{c_{k+i}}^R$ contains $t' - \ell$ $\bbK$-qudits. Furthermore,  $[R_1 | R_2] \in \bbK^{(t' - \delLoss) \times (d + 1)}$ is an interpolation matrix for a polynomial of degree $p$. By \Cref{fact:interpolation-linear-independence}, any $t' - \delLoss$ columns of $[R_1 | R_2]$ are linearly independent. In particular, any $t' - \delLoss$ columns of $R_1$ are linearly independent. Finally, note that $(t' - (t'-\ell))/2 = \ell/2$. Therefore by \Cref{lem:extractor}, the state of the system after preparing register $\chalReg_{c_{k+i}}^R$ and tracing out $\chalReg_{d_{i}}$ when the experiment ends at the end of this round is
     \[
        \sum_{\vect{x_3} \in \bbK^{t' - \delLoss}} \Tr^{\chalReg_{d_{i}}} \left[\ket{\gamma_{\vect{x_3}}}\bra{\gamma_{\vect{x_3}}}^{\chalReg,\calS,\calR} \right]
    \]
    where
    \[
        \ket{\gamma_{\vect{x_3}}} = 
        \sum_{\substack{\vect{x_1} \in \bbK^{t'}\\\vect{x_2} \in \bbK^{d - t'} }}  \alpha_{\vect{x_1}, \vect{x_2}} \ket{\vect{x_1}}^{\chalReg_{d_{i}}} 
        \otimes \ket{\vect{x_2}}^{\calI} 
        \otimes \ket{\vect{x_3} + \vect{R_2} \vect{x_2}}^{\chalReg_{c_{k+i}}} 
        \otimes \ket{\vect{x_3} + \vect{R_2} \vect{x_2}}^{\calS_{c_{k+i}}} 
        \otimes \ket{\varphi_{\vect{x_1}, \vect{x_2}}}^{\chalReg', \calS', \calR}
    \]

    After this round, $\calH_{3 + i}[c_{k+i}]$ ends and register $\chalReg$ is traced out. This yields the state
    \begin{align*}
        &\sum_{\vect{x_3} \in \bbK^{t' - \delLoss}}\Tr^{\chalReg}\left[\ket{\gamma_{\vect{x_3}}}\bra{\gamma_{\vect{x_3}}}^{\chalReg, \calS, \calR} \right]
        \\
        &\quad= \sum_{\substack{\vect{x_1} \in \bbK^{t'}\\\vect{x_2} \in \bbK^{d - t'} \\ \vect{x_3} \in \bbK^{t' - \delLoss}}} 
            |\alpha_{\vect{x_1}, \vect{x_2}}|^2 \ket{\vect{x_3} + \vect{R_2}\vect{x_2}}\bra{\vect{x_3} + \vect{R_2}\vect{x_2}}^{\calS_{c_{k+i}}^R} 
            \otimes \Tr^{C'}\left[\ket{\phi_{\vect{x_1}, \vect{x_2}}}\bra{\phi_{\vect{x_1}, \vect{x_2}}}^{\chalReg', \calS', \calR}\right]
        \\
        &\quad= \sum_{\substack{\vect{x_1} \in \bbK^{t'}\\\vect{x_2} \in \bbK^{d - t'} \\ \vect{x_4} \in \bbK^{t' - \delLoss}}} 
            |\alpha_{\vect{x_1}, \vect{x_2}}|^2 \ket{\vect{x_4}}\bra{\vect{x_4}}^{\calS_{c_{k+i}}^R} 
            \otimes \Tr^{\chalReg'}\left[\ket{\phi_{\vect{x_1}, \vect{x_2}}}\bra{\phi_{\vect{x_1}, \vect{x_2}}}^{\chalReg', \calS', \calR}\right]
    \end{align*}
    where $\vect{x_4} = \vect{x_3} + \vect{R_2}\vect{x_2}$. This state is identical to the state at the end of $\calH'_{3 + i}[c_{k+i}]$ conditioned on the experiments not aborting.
\ifsubmission\qed\fi\end{proof}

\begin{proposition} \label{claim:h2pi-to-h3pi}
    For every $i\in [0,n-k]$ and every secret $s$,
    \[
    \tracedist[\calH_{2+i}(s), \calH_{3 + i}(s)] = \negl
    \]
\end{proposition}
\begin{proof}
    The only difference between $\calH_{2+i}(s)$ and $\calH_{3 + i}(s)$ is an additional deletion predicate measurement $\Pi_{\cert_{d_{i+1}}}$ during the round where $d_{i+1}$ is corrupted. Say the deletion predicate accepts with probability $1-\epsilon$. Then the Gentle Measurement Lemma (\Cref{lem:gentle-measurement}) implies that, conditioned on the deletion predicate accepting, the distance between $\calH_{2+i}(s)$ and $\calH_{3 + i}(s)$ is at most $2\sqrt{\epsilon}$. We upper bound the case where the deletion predicate rejects by $1$ to obtain
    \[
        \tracedist[\calH_{2+i}(s), \calH_{3 + i}(s)] \leq (1-\epsilon)2\sqrt{\epsilon} + \epsilon
    \]
    
    Thus, it is sufficient to show that $\epsilon = \negl$, i.e. the deletion predicate accepts with high probability on $\chalReg_{d_{i+1}}$ in $\calH_{3+i}$.
    To show this, we consider the following hybrids, and claim that the probability that the deletion predicate accepts on $\chalReg_{d_{i+1}}$ is \emph{identical} in each of the hybrids.
    \begin{itemize}
        \item $\calH_{3+i}$
        \item $\calH_{3+i}[d_{i+1}]$: The only difference is that the game ends after the round where $d_{i+1}$ is deleted.
        \item $\Sim_{i+1}[d_{i+1}]$: Recall that the only difference between $\calH_{3+i}$ and $\Sim_{i+1}$ is that every share $j$ is prepared as the maximally entangled state
        \[
            \sum_{\vect{x} \in \bbK^{t}} \ket{\vect{x}}^{C_j} \otimes \ket{\vect{x}}^{A_j}
        \]
        
        \item $\Sim'_{i+1}[d_{i+1}]$: The same as $\Sim_{i+1}[d_{i+1}]$, except that after preparing the maximally entangled state for each share $j$, we delay choosing $J_j$ and measuring the check indices $\chalReg_{i,j}$ for $j\in [t]\backslash J_j$. These are now done immediately after $\Adv$ deletes $j$ by outputting $(\cert_j, j, \calR)$, and before the challenger verifies $\cert_j$.
    \end{itemize}

    Observe that $\calH_{3+i}$ and $\calH_{3+i}[d_{i+1}]$ are identical until the deletion predicate measurement in the round where $d_{i+1}$ is deleted, so the probability of acceptance is identical. By \Cref{claim:truncated-is-sim}, 
    \[
        \tracedist[\calH_{3+i}[c_{k+i}](s),\Sim_{i+1}[c_{k+i}]] = 0
    \]
    Share $d_{i+1}$ is deleted before share $c_{k+i}$ is corrupted in both $\calH_{3+i}(s)$ and $\Sim_{i+1}$, unless they abort. Therefore 
    \[
        \tracedist[\calH_{3+i}[d_{i+1}](s),\Sim_{i+1}[d_{i+1}]] = 0
    \]
    and the probability of acceptance is identical in $\calH_{3+i}[d_{i+1}](s)$ and $\Sim_{i+1}[d_{i+1}]$. Finally, the probability of acceptance is identical in $\Sim'_{i+1}[d_{i+1}]$ because the register $C$ is disjoint from the adversary's registers. 

    Thus, it suffices to show that $\epsilon = \negl$ in $\Sim'_{i+1}[d_{i+1}]$. Since $\bbK$ forms a vector space over $\bbF_2$, the certificate verification measurement and $\Pi_{\cert_{d_{i+1}}}$ are diagonal in the binary Fourier basis (i.e. the Hadamard basis) for every $\cert$. Therefore the probability that $\Verify$ accepts $\cert_{d_{i+1}}$ but the the deletion predicate measurement \emph{rejects} $\chalReg_{d_{i+1}}$ is
    \[
        \epsilon = \Pr_{\substack{\vect{\cert}, \vect{y} \in \bbK^{t}\\ J\subset [t]: |J| = t'}}\left[\vect{\cert}_{\overline{J}} = \vect{y}_{\overline{J}} \land \Delta_{\bbK}(\vect{\cert}_{J}, \vect{y}_{J}) \geq \frac{\delLoss}{2}\right]
    \]
    where $J$ is the set of data indices for share $d_{i+1}$, where $\overline{J}$ is the set complement of $J$ (i.e. the set of check indices for share $d_{i+1}$), and where $\Delta_{\bbK}(\vect{\cert}_{J}, \vect{y}_{J}) = h_{\bbK}(\vect{\cert}_{J} - \vect{y}_{J})$ is the Hamming distance of $\cert_{J}$ from $y_{J}$. Here, the probability is taken over the adversary outputting a certificate $\cert$ for $d_{k+i}$, the challenger sampling a set of check indices $\overline{J}$, and the challenger measuring all of register $\chalReg_{d_{i+1}}$ in the Hadamard basis to obtain $\vect{y}\in \bbK^t$.
    
    This value can be upper bounded using Hoeffding's inequality, for any fixed $\cert$ and $\vect{y}$ with $\Delta_{\bbK}(\cert_{J}, \vect{y}_{J}) \geq \delLoss/2$. Note that the probability of acceptance is no greater than if the $\numChecks$ check indices $\overline{J}$ are sampled \emph{with} replacement. In this case, the expected number of check indices which do \emph{not} match is 
    \begin{align}
        \geq \frac{\delLoss\numChecks}{2t} 
        &= \frac{t\log(\secpar)}{\secpar + (n-k+1)\log(\secpar)} \frac{(\secpar + (n-k+1)\log(\secpar))^2}{2t} 
        \\
        &= \frac{\log(\secpar)}{2} (\secpar + (n-k+1)\log(\secpar))
    \end{align}
    Therefore Hoeffding's inequality (\Cref{claim:hoeffding}) implies that
    \begin{align}
        \epsilon 
        &\leq 2\exp\left(
            \frac{-2\left(
                    \frac{\log(\secpar)}{2} (\secpar + (n-k+1)\log(\secpar))
                \right)^2}
                {(\secpar + (n-k+1)\log(\secpar))^2}
        \right) 
        \\
        &= 2\exp\left(-\frac{\log^2(\secpar)}{2}\right) 
        \\
        &= \negl
    \end{align}

    
\ifsubmission\smallskip\qed\fi\end{proof}
\fi
\fi 

\end{document}